\documentclass{article}
\usepackage{amsmath,amssymb,amsthm}
\usepackage{stmaryrd}
\usepackage[english]{babel}
\usepackage[marginparwidth=1in]{geometry}
\usepackage{graphicx}
\usepackage{color}
\usepackage{tikz}
\title{Fundamental performance limits for ideal decoders in high-dimensional linear inverse problems}
\author{Anthony Bourrier \and  Mike E. Davies \and Tomer Peleg \and Patrick P\'erez \and R\'emi Gribonval}

\newtheorem{theorem}{Theorem}
\newtheorem{lemma}{Lemma}
\newtheorem{proposition}{Proposition}
\newtheorem{remark}{Remark}

\begin{document}

\def\d{\Delta}
\def\dd{\Delta_{\delta}}
\def\deltap{\Delta^{\prime}}
\def\dde{\Delta_{\delta,\eps}}
\def\ds{d_{\Sigma}}
\def\e{\mathbb{E}}
\def\n{\mathbb{N}}
\def\r{\mathbb{R}}
\def\rd{\r^d}
\def\rn{\r^n}
\def\rm{\r^m}
\def\rk{\r^k}
\def\lo{L^1(\rn)}
\def\ba{\mbf{a}}
\def\be{\mbf{e}}
\def\bf{\mbf{f}}
\def\bg{\mbf{g}}
\def\bh{\mbf{h}}
\def\balpha{\boldsymbol\alpha}
\def\bbeta{\boldsymbol\beta}
\def\bgamma{\boldsymbol\gamma}
\def\bxi{\boldsymbol\xi}
\def\bm{\boldsymbol\mu}
\def\bn{\boldsymbol\nu}
\def\bo{\boldsymbol\omega}
\def\bop{\bo^{\prime}}
\def\bp{\mbf{p}}
\def\bq{\mbf{q}}
\def\bu{\mbf{u}}
\def\bv{\mbf{v}}
\def\bw{\mbf{w}}
\def\bx{\mbf{x}}
\def\bxp{\mbf{x^{\prime}}}
\def\by{\mbf{y}}
\def\bz{\mbf{z}}
\def\bxo{\bx_{opt}}
\def\bxs{\bx_{\Sigma}}
\def\bA{\mbf{A}}
\def\bB{\mbf{B}}
\def\bC{\mbf{C}}
\def\bD{\mbf{D}}
\def\bE{\mbf{E}}
\def\bI{\mbf{I}}
\def\Bt{\mathcal{B}_2}
\def\bL{\mbf{L}}
\def\bS{\mbf{S}}
\def\bM{\mbf{M}}
\def\bN{\mbf{N}}
\def\bU{\mbf{U}}
\def\bV{\mbf{V}}
\def\bX{\mbf{X}}
\def\bMt{\wtilde{\bM}}
\def\bP{\mbf{P}}
\def\pn{p_{\N}}
\def\pnt{p_{\Nt}}
\def\pno{p_{\No}}
\def\B{\mathcal{B}}
\def\H{\mathcal{H}}
\def\N{\mathcal{N}}
\def\Nt{\wtilde{\N}}
\def\X{\mathcal{X}}
\def\No{\mathcal{N}^{\perp}}
\def\P{\mathcal{P}}
\def\S{\mathcal{S}}
\def\Sigmap{\Sigma^{\prime}}
\def\sms{\Sigma-\Sigma}
\def\eps{\epsilon}
\def\im{\mathrm{im}}
\def\cm{\mathbb{C}^m}
\def\om{\omega}
\def\ie{\emph{i.e.}}
\def\iid{\emph{i.i.d.}}
\def\wrt{\emph{w.r.t. }}

\newcommand{\mbf}[1]{\ensuremath{\mathbf{#1}}}
\newcommand{\dy}[1]{\ensuremath{d_E(#1)}}
\newcommand{\dz}[1]{\ensuremath{d_F(#1)}}
\newcommand{\dM}[1]{\ensuremath{d_{M}(#1)}}
\newcommand{\norm}[1]{\ensuremath{\|#1\|}}
\newcommand{\enorm}[1]{\ensuremath{\|#1\|_E}}
\newcommand{\fnorm}[1]{\ensuremath{\|#1\|_F}}
\newcommand{\gnorm}[1]{\ensuremath{\|#1\|_G}}
\newcommand{\spnorm}[1]{\ensuremath{\|#1\|_{\Sigmap}}}
\newcommand{\sknorm}[1]{\ensuremath{\|#1\|_{\Sigma_{k}}}}
\newcommand{\snorm}[1]{\ensuremath{\|#1\|_{\Sigma}}}
\newcommand{\tracenorm}[1]{\ensuremath{\|#1\|_*}}
\newcommand{\nx}[1]{\ensuremath{\|#1\|_G}}
\newcommand{\ny}[1]{\ensuremath{\|#1\|_E}}
\newcommand{\nz}[1]{\ensuremath{\|#1\|_F}}
\newcommand{\nfrob}[1]{\ensuremath{\|#1\|_F}}
\newcommand{\nzero}[1]{\ensuremath{\|#1\|_0}}
\newcommand{\ntwo}[1]{\ensuremath{\|#1\|_2}}
\newcommand{\none}[1]{\ensuremath{\|#1\|_1}}
\newcommand{\ninf}[1]{\ensuremath{\|#1\|_{\infty}}}
\newcommand{\nM}[1]{\ensuremath{\|#1\|_{M}}}
\newcommand{\dec}[1]{\ensuremath{\Delta(#1)}}
\newcommand{\ps}[2]{\ensuremath{\langle #1,#2\rangle}}
\renewcommand{\inf}[2]{\ensuremath{\underset{#1}{\mathrm{inf}}\quad #2}}
\renewcommand{\sup}[2]{\ensuremath{\underset{#1}{\mathrm{sup}}\quad #2}}
\newcommand{\argmin}[2]{\ensuremath{\underset{#1}{\mathrm{argmin}}\quad #2}}
\newcommand{\wtilde}[1]{\ensuremath{\widetilde{#1}}}

\newcommand{\revision}[1]{#1}

\newcommand\blfootnote[1]{%
  \begingroup
  \renewcommand\thefootnote{}\footnote{#1}%
  \addtocounter{footnote}{-1}%
  \endgroup
}

\maketitle

\blfootnote{A. Bourrier is with Gipsa-Lab. M.E. Davies is with University of Edimburgh. T. Peleg is with Israel Institute of Technology. P. P\'erez is with Technicolor. R. Gribonval is with INRIA.}

\begin{abstract} 
The primary challenge in linear inverse problems is to design stable and robust ``decoders'' to reconstruct high-dimensional vectors from a low-dimensional observation through a linear operator. 
Sparsity, low-rank, and related assumptions are typically exploited to design decoders which performance is then bounded based on some measure of deviation from the idealized model, typically using a norm. 
 
This paper focuses on characterizing the fundamental performance limits that can be expected from an ideal decoder given a general model, \ie, a general subset of ``simple'' vectors of interest. 
First, we extend the so-called notion of instance optimality of a decoder to settings where one only wishes to reconstruct some part of the original high dimensional vector from a low-dimensional observation. 
This covers practical settings such as medical imaging of a region of interest, or audio source separation when one is only interested in estimating the contribution of a specific instrument to a musical recording. 
We define instance optimality relatively to a model much beyond the traditional framework of sparse recovery, and characterize the existence of an instance optimal decoder in terms of joint properties of the model and the considered linear operator. 
Noiseless and noise-robust settings are both considered. We show somewhat surprisingly that the existence of {\em noise-aware} instance optimal decoders for all noise levels implies the existence of a {\em noise-blind} decoder. 
 
A consequence of our results is that for models that are rich enough to contain an orthonormal basis, the existence of an $\ell^{2}/\ell^{2}$ instance optimal decoder is only possible when the linear operator is not substantially dimension-reducing.
This covers well-known cases (sparse vectors, low-rank matrices) as well as a number of seemingly new situations (structured sparsity and sparse inverse covariance matrices for instance).

We exhibit an operator-dependent norm which, under a model-specific generalization of the Restricted Isometry Property (RIP), always yields a feasible instance optimality property. \revision{This norm can be upper bounded by an atomic norm relative to the considered model.}

\end{abstract}

\section{Introduction}

In linear inverse problems, one considers a linear measurement operator $\bM$ mapping the signal space $\rn$ to a measurement space $\rm$, where typically $\bM$ is either ill-conditioned or dimensionality reducing. The reconstruction of $\bx$ from $\bM \bx$ is thus a hopeless task unless one can exploit prior knowledge on $\bx$ to complete the incomplete observation $\bM \bx$.   

Sparsity is a well-known enabling model for this type of inverse problems: it has been proven that \revision{for certain such operators $\bM$,} one can expect to recover the signal $\bx$ from its measure $\bM\bx$ provided that $\bx$ is sufficiently sparse, \ie, it has few nonzero components \cite{CandesRT06}. If the set of $k$-sparse signals is denoted $\Sigma_k=\{\bx\in\rn,\nzero{\bx}\leq k\}$, where $\nzero{.}$ is the pseudo-norm counting the number of nonzero components, then this recovery property can be interpreted as the existence of a decoder $\Delta: \rm\rightarrow\rn$ such that $\forall\bx\in\Sigma_k,\  \Delta(\bM\bx)=\bx$, thus making $\bM$ a linear encoder associated to the (typically nonlinear) decoder $\Delta$.

Further, the body of theoretical work around sparse recovery in linear inverse problems has given rise to the notion of compressive sensing (CS) \cite{Donoho06}, where the focus is on {\em choosing} \revision{-- among a more or less constrained set of operators --} a dimensionality reducing $\bM$ to which a decoder can be associated\footnote{\revision{By contrast, linear inverse problems usually refer to a setup where one aims at reconstructing a signal from its measurements by a given operator (\emph{e.g.} imposed by the underlying physics), which may be dimensionality-reducing.}}. It is now well-established that this can be achieved 
in scenarii where $m\ll n$, showing that a whole class of seemingly high-dimensional signals can thus be reconstructed from far lower dimensional linear measurements than their apparent dimension.

\subsection{Instance optimal sparse decoders}
A good decoder $\Delta$ is certainly expected to have nicer properties than simply reconstructing $\Sigma_k$. Indeed, the signal $\bx$ to be reconstructed may not belong exactly in $\Sigma_k$ but ``live near'' $\Sigma_k$ under a distance $d$, meaning that $d(\bx,\Sigma_k)$ is ``small'' in a certain sense. In this case, one wants to be able to build a sufficiently precise estimate of $\bx$ from $\bM\bx$, that is a quantity $\Delta(\bM\bx)$ such that $\|\bx-\Delta(\bM\bx)\|$ is ``small'' for a certain norm $\|.\|$. This stability to the model has been formalized into the so-called \emph{Instance Optimality} assumption on $\Delta$. Decoder $\Delta$ is said to be instance optimal if:
\begin{equation}
\label{schematic_io}
\forall\bx\in\rn, \|\bx-\Delta(\bM\bx)\|\leq Cd(\bx,\Sigma_k),
\end{equation}
for a certain choice of norm $\|.\|$ and distance $d$. For this property to be meaningful, the constant $C$ must not scale with $n$ and typically ``good'' instance optimal decoders are decoders which involve a constant which is the same for all $n$ (note that this implicitly relies on the fact that a sparse set $\Sigma_{k} \subset \rn$ can be defined for any $n$). When the norm is $\ell^2$ or $\ell^1$ and the distance is $\ell^1$, such good instance optimal decoders exist and can be implemented as the minimization of a convex objective \cite{Donoho06,CandesT06,Candes08} under assumptions on $\bM$ such as the Restricted Isometry Property (RIP). Note that instance optimality is a uniform upper bound on the reconstruction error, and that other types of bounds on decoders can be studied, particularly from a probabilistic point of view \cite{ChandrasekaranRPW12}. Other early work include upper bounds on the reconstruction error from noisy measurements with a regularizing function when the signal belongs exactly to the model \cite{EnglHN96}.

In \cite{Cohen09compressedsensing}, the authors considered the following question: \textit{Given the encoder $\bM$, is there a simple characterization of the existence of an instance optimal decoder?} Their goal was not to find implementable decoders that would have this property, but rather to identify conditions on $\bM$ and $\Sigma_k$ under which the reconstruction problem is ill-posed if one aims at finding an instance optimal decoder with small constant. The existence of a decoder $\Delta$ which satisfies \eqref{schematic_io} will be called the \emph{Instance Optimality Property} (IOP). The authors proved that this IOP is closely related to a property of the kernel of $\bM$ with respect to $\Sigma_{2k}$, called the \emph{Null Space Property} (NSP). This relation allowed them to study the existence of stable decoders under several choices of norm $\|.\|$ and distance $d(.,.)$. 

A related question addressed in \cite{Cohen09compressedsensing} is that of the fundamental limits of dimension reduction: \textit{Given the target dimension $m$ and desired constant $C$, is there an encoder $\bM$ with an associated instance optimal decoder?} They particularly showed that there is a fundamental trade-off between the size of the constant $C$ in \eqref{schematic_io} (with $\ell^2$ norm and $\ell^2$ distance) and the dimension reduction ratio $m/n$.

\subsection{Low-dimensional models beyond sparsity}
Beyond the sparse model, many other low-dimensional models have been considered in the context of linear inverse problems and CS \cite{BaraniukCW10}. In these generalized models, the signals of interest typically live in or close to a subset $\Sigma$ of the space, which typically contains far fewer vectors than the whole space. Such models encompass sets of elements as various as block-sparse signals \cite{EldarKB10}, unions of subspaces, whether finite \cite{BlumensathD09} or possibly infinite \cite{Blumensath11}, signals sparse in a redundant dictionary \cite{RauhutSV08}, cosparse signals \cite{Nam2013}, approximately low-rank matrices \cite{RechtFP10,Candes2011}, low-rank and sparse matrices \cite{ZhouLWCM10,CandesLMW11}, symmetric matrices with sparse inverse \cite{Yuan2007,Yuan2010} or manifolds \cite{Baraniuk06randomprojections,EftekhariW13}. An old result which can also be interpreted as generalized CS is the low-dimensional embedding of a point cloud \cite{johnson84extensionslipschitz,Ach01}. Some of these models are pictured in Figure \ref{fig:model_illust}.

\begin{figure}
\centering
\includegraphics[scale=.25]{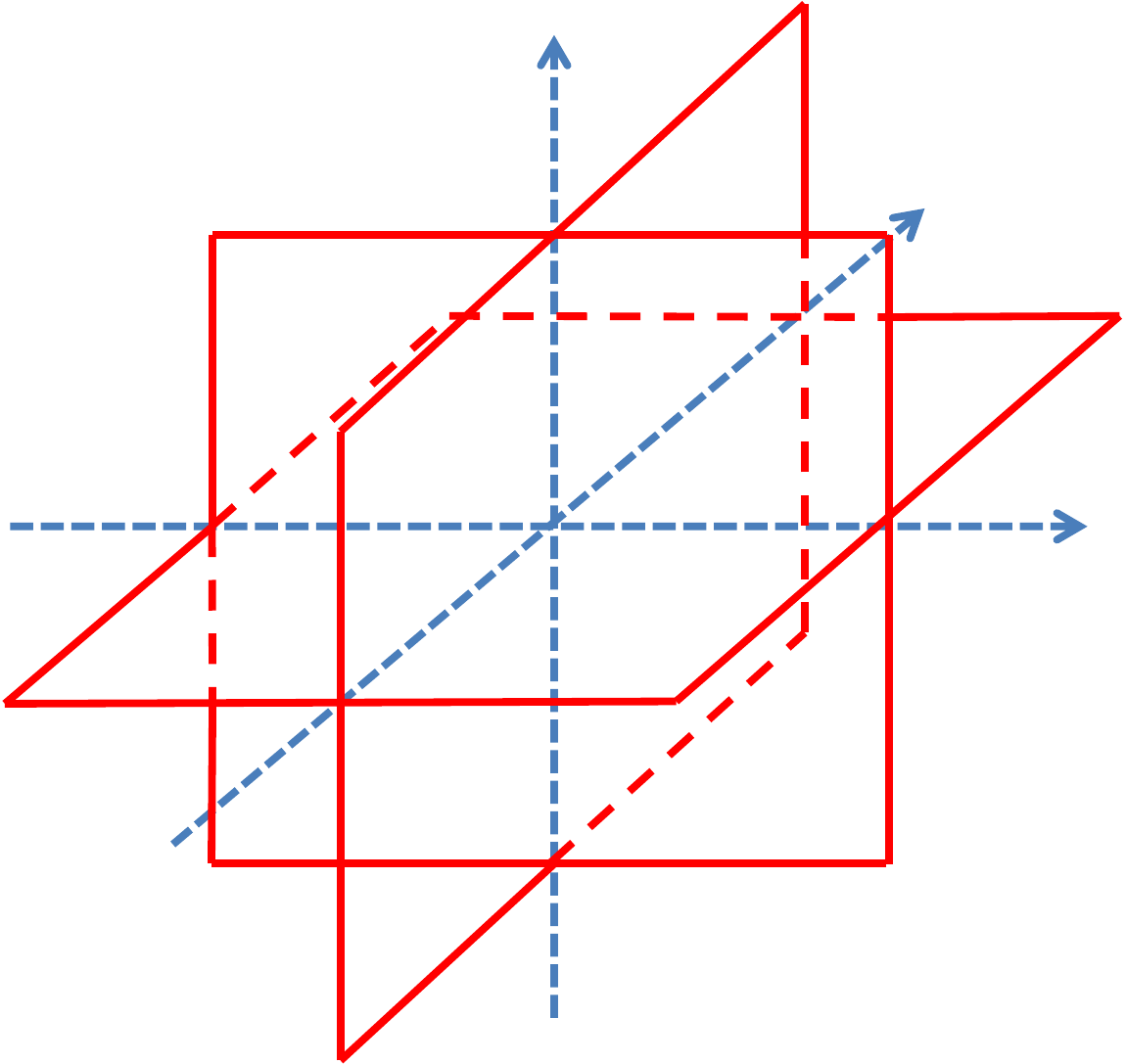}
\hspace{2em}
\includegraphics[scale=.25]{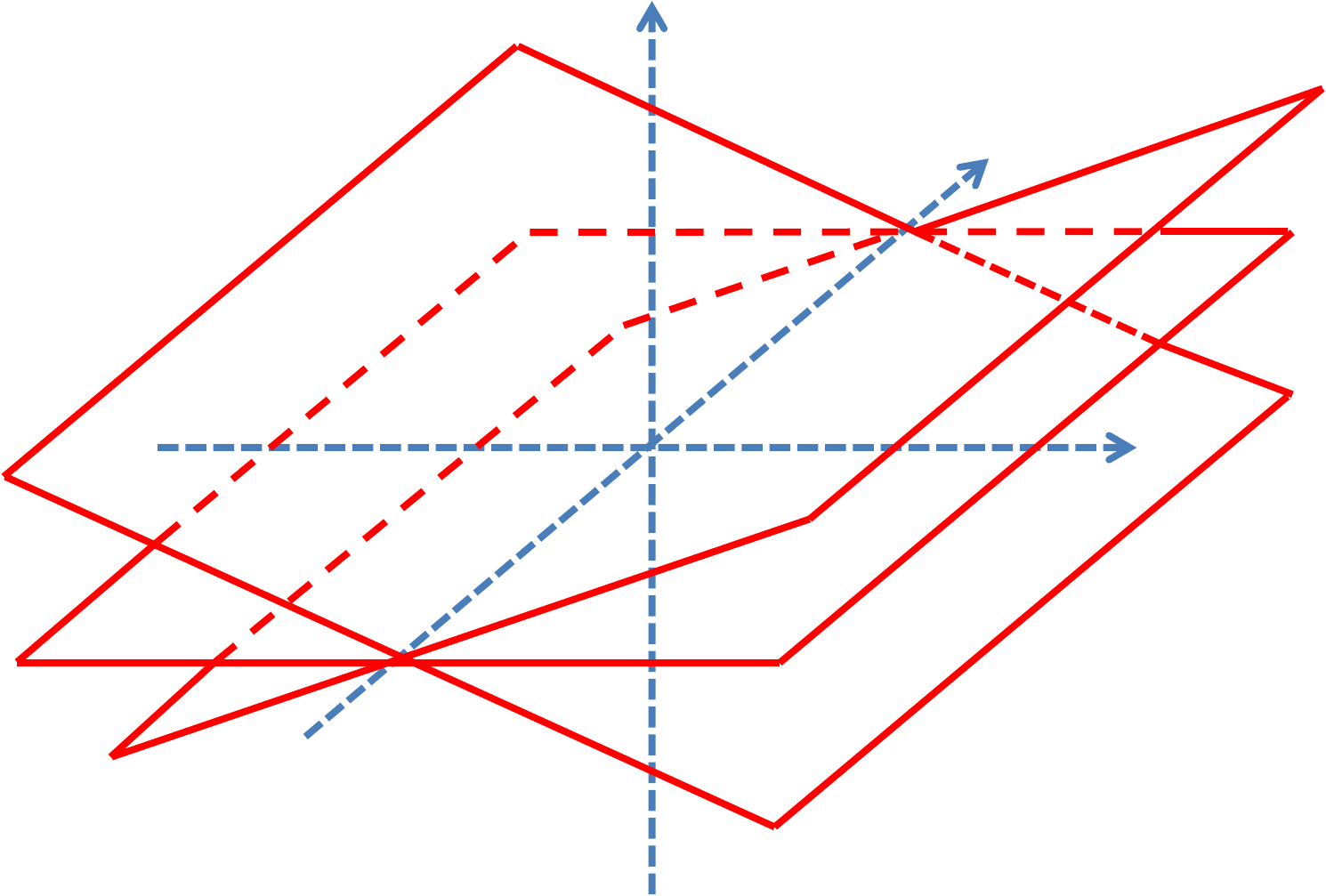}
\hspace{2em}
\includegraphics[scale=.25]{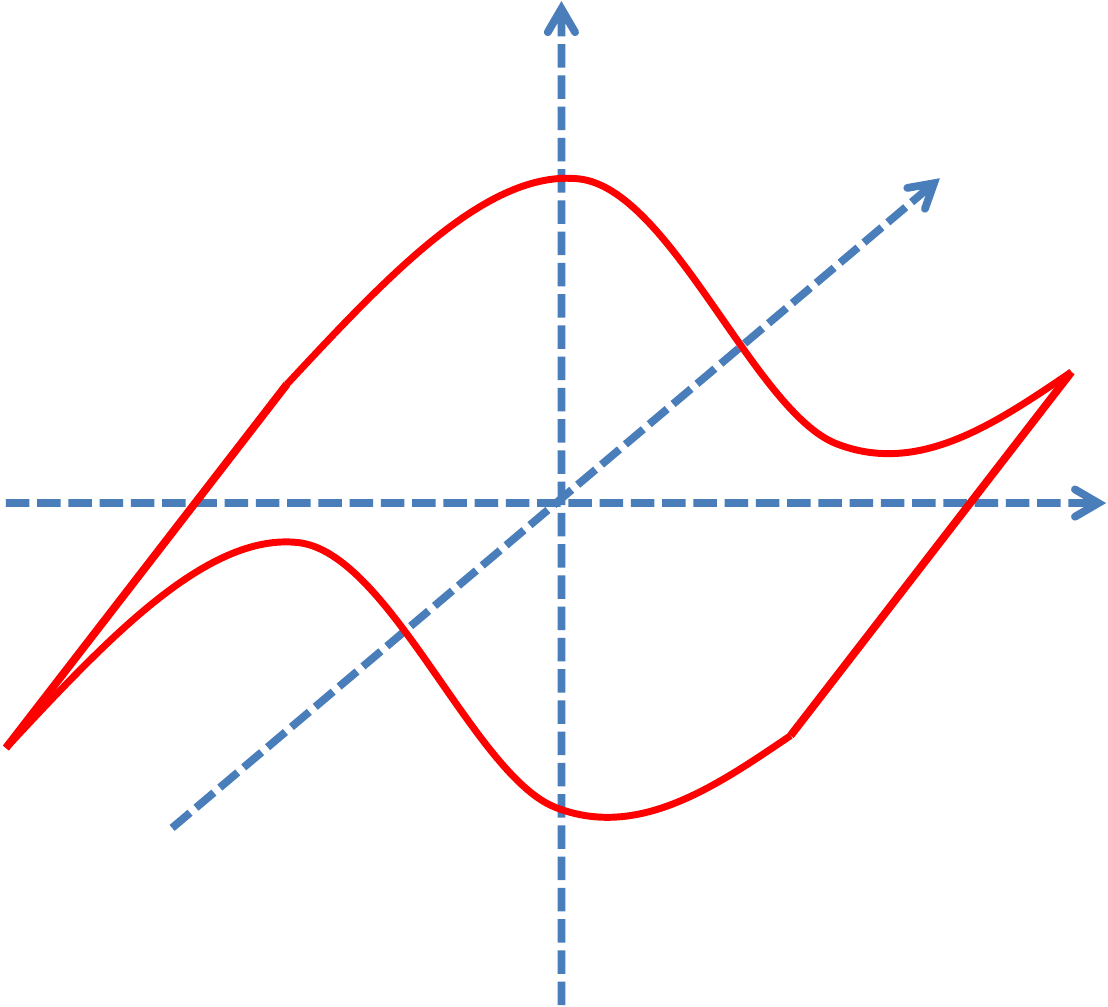}
\hspace{2em}
\includegraphics[scale=.25]{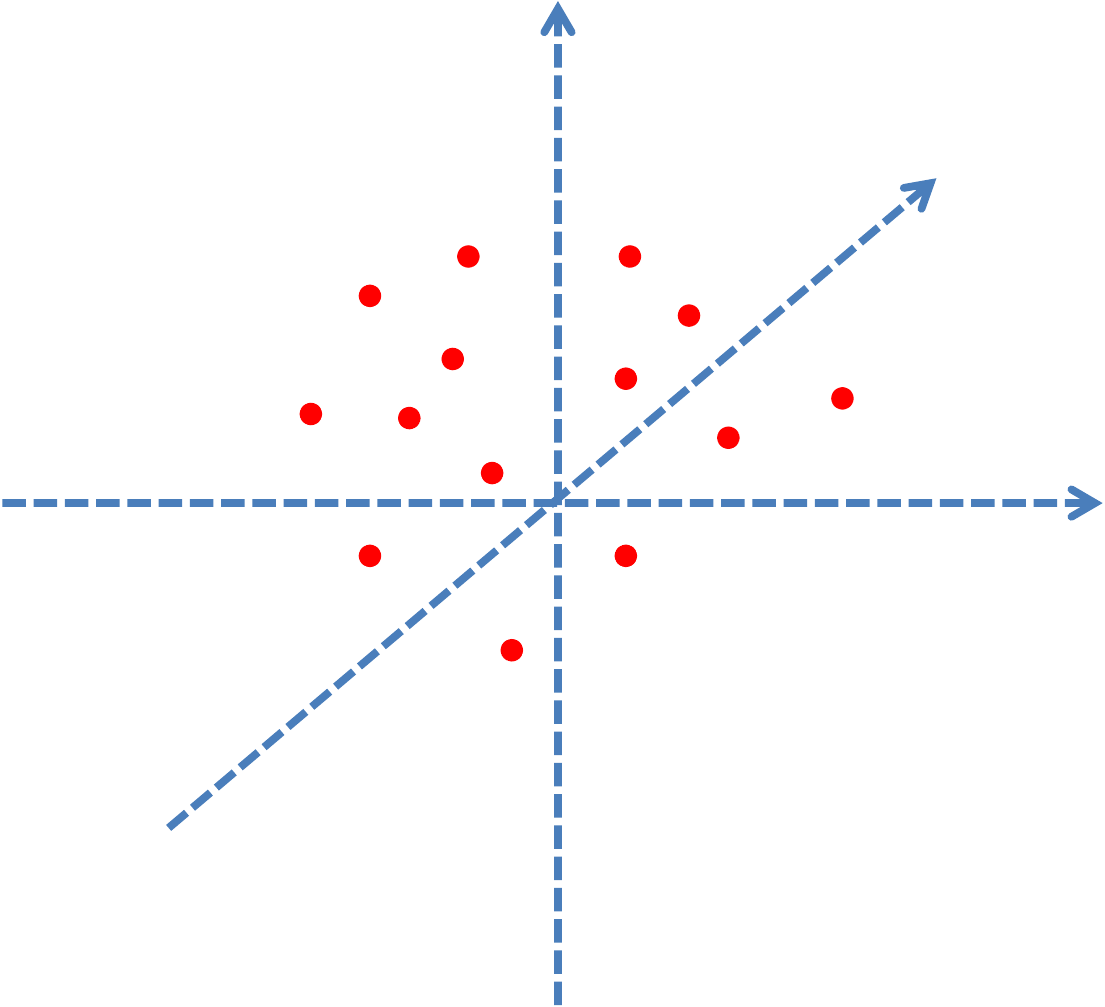}
\caption{Illustration of several CS models. \emph{From left to right:} $k$-sparse vectors, union of subspaces, smooth manifold and point cloud.}
\label{fig:model_illust}
\end{figure}

Since these models generalize the sparse model, the following question arises: can they be considered under a general framework, sharing common reconstruction properties? In this work, we are particularly interested in the extension of the results of \cite{Cohen09compressedsensing} to these general models, allowing to further investigate the well-posedness of such problems.

In \cite{PelegGD13}, the theoretical results of \cite{Cohen09compressedsensing} are generalized in the case where one aims at stably decoding a vector living near a finite union of subspaces (UoS). They also show in this case the impossibility of getting a good $\ell^2/\ell^2$ instance optimal decoder with substantial dimensionality reduction. Their extension also covers the case where the quantity one wants to decode is not the signal itself but a linear measure of the signal.

In this work, we further extend the study of the IOP to general models of signals: we consider signals of interest living in or near a subset $\Sigma$ of a vector space $E$, without further restriction, and show that instance optimality can be generalized for such models. In fact, we consider the following generalizations of instance optimality as considered in \cite{Cohen09compressedsensing}:
\begin{itemize}
\item \textbf{Robustness to noise:} noise-robust instance optimality is characterized, showing somewhat surprisingly the equivalence between the existence of two flavors of noise-robust decoders ({\em noise-aware} and {\em noise-blind});
\item \textbf{Infinite dimension:} signal spaces $E$ that may be infinite dimensional are considered. For example $E$ may be a Banach space such as an $L^p$ space or a space a signed measures. This is motivated by recent work on infinite dimensional compressed sensing \cite{BenAdcock:2013va} or compressive density estimation \cite{Bourrier:2013wl};
\item \textbf{Task-oriented decoders:} the decoder is not constrained to approximate the signal $\bx$ itself but rather a linear feature derived from the signal, $\bA\bx$, as in \cite{PelegGD13}; in the usual inverse problem framework, $\bA$ is the identity. Examples of problems where $\bA\ne\bI$ include:
\begin{itemize}
\item Medical imaging of a particular region of the body: as in Magnetic Resonance Imaging, one may acquire Fourier coefficients of a function defined on the body, but only want to reconstruct properly a particular region. In this case, $\bA$ would be the orthogonal projection on this region.
\item Partial source separation: given an audio signal mixed from several sources whose positions are known, as well as the microphone filters, the task of isolating one of the sources from the mixed signal is a reconstruction task where $E$ is the space of concatenated sources, and $\bA$ orthogonally projects such a signal in a single source signal space.
\end{itemize}
%A typical example is brain imaging, where $\bM\bx$ encodes information from the whole brain while one only wants to accurately reconstruct in a certain region of interest.
\item \revision{\textbf{Pseudo-norms:} Instead of considering instance optimality involving norms, we use pseudo-norms with fewer constraints, allowing us to characterize a wider range of instance optimality properties. As we will see in Section \ref{sec:task_io}, this flexibility on the pseudo-norms has a relationship with the previous point: it essentially allows one to suppose $\bA=\bI$ in every case, up to a change in the pseudo-norm considered for the approximation error.}
\end{itemize}

%\subsection{Summary of the main contributions}
\subsection{\revision{Contributions of this work}}

\revision{We summarize below our main contributions.}

\subsubsection{Instance optimality for inverse problems with general models}

In the noiseless case, we express a concept of instance optimality which does not necessarily involve homogeneous norms and distances but some pseudo-norms instead. 
Such a generalized instance optimality can be expressed as follows:
\begin{equation}
\label{io_generalized}
\forall\bx\in E, \nx{\bA\bx-\Delta(\bM\bx)}\leq C d_E(\bx,\Sigma),
\end{equation}
where $\nx{.}$ is a pseudo-norm and $d_E$ is a distance the properties of which will be specified in due time, and $\bA$ is a linear operator representing the feature one wants to estimate from $\bM\bx$. 
Our first contribution is to prove that the existence of a decoder $\Delta$ satisfying \eqref{io_generalized}, which is a generalized IOP, can be linked with a generalized NSP, similarly to the sparse case. This generalized NSP can be stated as:
\begin{equation}
\label{nsp_generalized}
\forall \bh\in\ker(\bM), \nx{\bA\bh}\leq D d_E(\bh,\sms),
\end{equation}
where the set $\sms$ is comprised of all differences of elements in $\Sigma$, that is $\sms=\{\bz_1-\bz_2:\bz_1,\bz_2\in\Sigma\}$. The constants $C$ and $D$ are related by a factor no more than 2, as will be stated in Theorems \ref{thm:io_nsp} and \ref{thm:nsp_io} characterizing the relationships between these two properties.
In particular, all previously mentioned low-dimensional models can fit in this generalized framework. 

\subsubsection{Noise-robust instance optimality}

Our second contribution (Theorems \ref{thm:io_nsp_noise} and \ref{thm:nsp_io_noise}) is to link a noise-robust extension of instance optimality to a property called the Robust NSP. Section \ref{sec:gen_io_nsp} regroups these noiseless and noise-robust results after a review of the initial IOP/NSP results of \cite{Cohen09compressedsensing}. 
We show somewhat surprisingly that the existence of {\em noise-aware} instance optimal decoders for all noise levels implies the existence of a {\em noise-blind} decoder (Theorem~\ref{thm:io_nsp_noise_aware}).

If a Robust NSP is satisfied, an instance optimal decoder can be defined as:
\begin{equation}
\label{eq:robust_decoder}
\Delta(\by)=\argmin{\bu\in E}{D_1 d_E(\bu,\Sigma)+ D_2 d_F(\bM\bu,\by)},
\end{equation}
where the constants $D_1,D_2$ and distances $d_E,d_F$ are those which appear in the Robust NSP. The objective function is the sum of two terms: a distance to the model and a distance to the measurements. Also note that by fixing $D_2$ to infinity, one defines an instance optimal noise-free decoder provided the corresponding NSP is satisfied.

\subsubsection{Limits of dimensionality reduction with generalized models}
\label{subsubsec:limits_dimred}
The reformulation of IOP as an NSP allows us to consider the $\ell^2/\ell^2$ instance optimality for general models in Section \ref{sec:l2l2}. In this case, the NSP can be interpreted in terms of scalar product and we precise the necessity of the NSP for the existence of an instance optimal decoder. This leads to the proof of Theorem \ref{thm:l2l2_io} stating that, just as in the sparse case, {\em one cannot expect to build an $\ell^2/\ell^2$ instance optimal decoder if $\bM$ reduces substantially the dimension and the model is ``too large''} in a precise sense.  In particular, we will see that the model is ``too large'' when the set $\Sigma-\Sigma$ contains an orthonormal basis. This encompasses a wide range of standard models where a consequence of our results is that $\ell^2/\ell^2$ IOP with dimensionality reduction is impossible:
\begin{itemize}
\item \textbf{$k$-sparse vectors}. In the case where $\Sigma=\Sigma_k$ is the set of $k$-sparse vectors, $\Sigma$ contains the null vector and the canonical basis, so that $\Sigma-\Sigma$ contains the canonical basis. Note that the impossibility of good $\ell^2/\ell^2$ IOP has been proved in \cite{Cohen09compressedsensing}.
\item \textbf{Block-sparse vectors} \cite{EldarKB10}. The same argument as above applies in this case as well, implying that imposing a block structure on sparsity does not improve $\ell^2/\ell^2$ feasibility.
\item \textbf{Low-rank matrices} \cite{RechtFP10,Candes2011}. In the case where $E=\mathcal{M}_n(\r)$ and $\Sigma$ is the set of matrices of rank $\leq k$, $\Sigma$ also contains the null matrix and the canonical basis.%, hence the impossibility of having a good $\ell^2/\ell^2$ IO decoder.
\item \textbf{Low-rank + sparse matrices} \cite{ZhouLWCM10,CandesLMW11}. The same argument applies to the case where the model contains all matrices that are jointly low-rank \emph{and} sparse, which appear in phase retrieval \cite{OymakJFEH12,OhlssonYDS11,CandesSV11}.
\item \textbf{Low-rank matrices with non-sparsity constraints}. In order to reduce the ambivalence of the low-rank + sparse decomposition of a matrix, \cite{CandesLMW11} introduced non-sparsity constraints on the low-rank matrix in order to enforce its entries to have approximately the same magnitude. However, as shown in Lemma \ref{lem:lr_nonsparse_mat}, an orthonormal Fourier basis of the matrix space can be written as differences of matrices which belong to this model.
\item \textbf{Reduced union of subspace models} \cite{BaraniukCW10} obtained by pruning out the combinatorial collection of $k$-dimensional subspaces associated to $k$-sparse vectors. This covers block-sparse vectors \cite{EldarKB10}, tree-structured sparse vectors, and more. Despite the fact that these unions of subspaces may contain much fewer $k$-dimensional subspaces than the combinatorial number of subspaces of the standard $k$-sparse model,  the same argument as in the $k$-sparse model applies to these signal models, provided they contain the basis collection of $1$-sparse signals. This contradicts the naive intuition that $\ell^{2}/\ell^{2}$ IOP could be achievable at the price of  substantially reducing the richness of the model through a drastic pruning of its subspaces.
\item \textbf{$k$-sparse expansions in a dictionary model} \cite{RauhutSV08}. More generally, if the model is the set of vectors which a linear combination of at most $k$ elements of a dictionary $\bD$ which contains an orthogonal family or a tight frame, then Theorem \ref{thm:l2l2_io} applies.
\item \textbf{Cosparse vectors with respect to the finite difference operator} \cite{Nam2013,PelegGD13}. As shown in \cite{PelegGD13}, the canonical basis is highly cosparse with respect to the finite difference operator, hence it is contained in the corresponding union of subspaces. 
\item As shown in Lemma \ref{lem:sdp_mat}, this is also the case for \textbf{symmetric definite positive square matrices with $k$-sparse inverse}. The covariance matrix of high-dimensional Gaussian graphical models is of this type: the numerous pairwise conditional independences that characterize the structure of such models, and make them tractable, translate into zeros entries of the inverse covariance matrix (the concentration matrix). Combining sparsity prior on the concentration matrix with maximum likelihood estimation of covariance from data, permits to learn jointly the structure and the parameters of Gaussian graphical models (so called ``covariance selection'' problem) \cite{Yuan2007,Yuan2010}. In very high-dimensional cases, compressive solutions to this problem would be appealing.
\item Johnson-Lindenstrauss embedding of \textbf{point clouds} \cite{Ach01}. Given a set $\mathcal{X}$ of $L$ vectors in $\rn$ and $\epsilon>0$, there exists a linear mapping $f:\rn\rightarrow \rm$, with $m=\mathcal{O}(\ln(L)/\eps^2)$ and 
\begin{equation}
(1-\eps)\ntwo{\bx-\by}\leq\ntwo{f(\bx)-f(\by)}\leq (1+\eps)\ntwo{\bx-\by}
\end{equation}
holds for all $\bx,\by\in \mathcal{X}$. The fact that the point cloud contains a tight frame is satisfied if it ``spreads'' in a number of directions which span the space. In this case, \emph{one cannot guarantee precise out-of-sample reconstruction} of the points in $\rn$ in the $\ell^2$-sense, except for a very limited neighborhood of the point cloud. This is further discussed in Section \ref{sec:discussion}.
\end{itemize}

\subsubsection{Generalized Restricted Isometry Property}
\label{sec:gen_rip}

Our last contribution, in Section \ref{sec:nsp_rip}, is to study the relations between the NSP and a generalized version of the Restricted Isometry Property (RIP). This generalized RIP bounds $\nz{\bM\bx}$ from below and/or above on a certain set $V$, and can be decomposed in:
\begin{align}
\label{eq:lower_rip}
\mathrm{Lower-RIP:} &\; \forall\bx\in V, \alpha\nx{\bx}\leq \nz{\bM\bx}\\
\label{eq:upper_rip}
\mathrm{Upper-RIP:} &\; \forall\bx\in V, \nz{\bM\bx}\leq \beta\nx{\bx},
\end{align}
where $\nx{.}$ and $\nz{.}$ are pseudo-norms defined respectively on the signal space and on the measure space, and $0<\alpha \leq \beta < +\infty$.
We prove particularly in Theorem \ref{thm:nsp_mnorm} that a generalized lower-RIP on $\sms$ implies the existence of instance optimal decoders in the noiseless and the noisy cases for a certain norm $\ny{\cdot}$ we call the ``$M$-norm''\footnote{The prefix ``M-'' should be thought as ``Measurement-related norm'' since in other works the measurement matrix may be denoted by other letters.}.

Furthermore, we prove that under an upper-RIP assumption on $\Sigma$, this $M$-norm can be upper bounded by an atomic norm \cite{ChandrasekaranRPW12} defined using $\Sigma$ and denoted $\snorm{.}$. This norm is easier to interpret than the $M$-norm: it can in particular be upper bounded by usual norms for the $k$-sparse vectors and low-rank matrices models. We have the following general result relating generalized RIP and IOP (Theorem \ref{thm:rip_nsp_full}): if $\bM$ satisfies a lower-RIP \eqref{eq:lower_rip} for $V=\sms$ and an upper-RIP \eqref{eq:upper_rip} for $V=\Sigma$, then for all $\delta>0$, there exists a decoder $\dd$ satisfying $\forall\bx\in E,\forall\be\in F$,
\begin{equation}
\nx{\bx-\dd(\bM\bx+\be)}\leq 2\left(1+\frac{\beta}{\alpha}\right) d_{\Sigma}(\bx,\Sigma)+\frac{2}{\alpha}\ny{\be}+\delta,
\end{equation}
which is a particular case of Robust instance optimality, as described in Section \ref{sec:gen_io_nsp}.

In particular, this generalized RIP encompasses classical or recent RIP formulations, such as
\begin{itemize}
\item The \textbf{standard RIP} \cite{Candes08} with $V$ as the set of $k$-sparse vectors, $\nx{\cdot}$ and $\nz{\cdot}$ being $\ell^2$ norms.
\item The \textbf{Union of Subspaces RIP} \cite{Blumensath11} with $V$ as a union of subspaces, $\nx{\cdot}$ and $\nz{\cdot}$ being $\ell^2$ norms.
\item The \textbf{RIP for low-rank matrices} \cite{Candes2011} with $V$ as the set of matrices of rank $\leq r$, $\nx{\cdot}$ as the Frobenius norm and $\nz{\cdot}$ as the $\ell^2$ norm;
\item The \textbf{$\bD$-RIP} \cite{candes2011csdico} for the dictionary model with $V$ as the set of vectors spanned by $k$ columns of a dictionary matrix, $\nx{\cdot}$ and $\nz{\cdot}$ being $\ell^2$ norms;
\item The \textbf{$\boldsymbol\Omega$-RIP} \cite{Giryes2012} for the cosparse model with $V$ as the set of vectors $\bx$ such that $\boldsymbol\Omega\bx$ is $k$-sparse, where $\Omega$ is the cosparse operator, $\nx{\cdot}$ and $\nz{\cdot}$ being $\ell^2$ norms;
\item Similarly, the \textbf{task-RIP} can be defined given a linear operator $\bA$ such that one aims at reconstructing the quantity $\bA\bx$ (instead of $\bx$) to perform a particular task. As we will see in Section \ref{sec:task_io}, in terms of IOP, this is essentially equivalent to reconstructing $\bx$ in terms of the norm $\nx{\bA\cdot}$. In this case, the corresponding lower task-RIP reads:
\begin{equation}
\alpha \nx{\bA\bx}\leq \nz{\bM\bx}.
\end{equation}
\end{itemize}

\revision{
\subsubsection{Infinite-dimensional inverse problems}

The generalization of the relationship between the IOP, the NSP and the RIP to arbitrary vector spaces allows us to consider recovery results for infinite dimensional inverse problems. Such problems have mainly been considered in separable Hilbert spaces \cite{Adcock12a,BenAdcock:2013va}, where the signals of interest are sparse with respect to a Hilbert basis and the measurement operators subsample along another Hilbert basis. In the theory of generalized sampling \cite{Adcock12a}, even when the signal model $\Sigma$ is simply a finite dimensional subspace, it can be necessary to oversample by some factor in order to guarantee stable recovery. In fact Theorem 4.1 of \cite{Adcock12b} can be read as a statement of $\ell^2/\ell^2$ instance optimality for a specific (linear) decoder given in terms of the NSP constant of the measurement operator. The results presented here therefore provide an extension of generalized sampling for linear models beyond $\ell^2$.

However, as mentioned in Section \ref{sec:sparse_hilbert}, \emph{one cannot hope to get uniform instance optimality in this setting for a standard sparsity model}. This is mentioned in \cite{BenAdcock:2013va} when the authors state that no RIP can be satisfied in this case. In section \ref{sec:topo_rip}, we nevertheless discuss the possibility of uniform instance optimality results in infinite dimensions \emph{with a proper choice of model $\Sigma$ and pseudo-norms}. In particular, a non-constructive topological result ensures that a generalized RIP is satisfied for a model of finite box-counting dimension \cite{Robinson10}. This generalized RIP leads to an IOP, according to Theorem \ref{thm:nsp_mnorm}.

Hopefully, our results will therefore help characterizing conditions under which infinite-dimensional uniform IOP is possible.}

%In the theory of generalized sampling \cite{Adcock12a}, even when the signal model $\Sigma$ is simply a finite dimensional subspace, it can be necessary to oversample by some factor in order to guarantee stable recovery. In fact Theorem 4.1 of \cite{Adcock12b} can be read as a statement of $\ell^2/\ell^2$ instance optimality for a specific (linear) decoder given in terms of the NSP constant of the measurement operator. The results presented here therefore provide an extension of generalized sampling for linear models beyond $\ell^2$.

%In \cite{Adcock12b} the authors also combine generalized sampling theory with compressed sensing. However, rather than 
%developing a uniform instance optimality theory which seems the most natural extension of generalized sampling, they 
%adopt a non-uniform approach based on \cite{CandesandPlan2010}. Our results should enable the development of a uniform 
%infinite dimensional CS theory.}

\subsection{Structure of the paper}
We will now describe the layout of the paper. Section \ref{sec:gen_io_nsp} first contains a quick review of the relationship between IOP and NSP in the usual sparse case, then exposes the more general setting considered in this paper, for which these properties and their relationship are extended, both in noiseless and noisy settings. Section \ref{sec:l2l2} then focuses on the particular case of $\ell^2/\ell^2$ IOP, proving the impossibility for a certain class of models to achieve such IOP with decent precision in dimension reducing scenarii. In particular, we show that this encompasses a wide range of usual models. Finally, in Section \ref{sec:nsp_rip}, we get back to the problem of IO with general norms and prove that a generalized version of the lower-RIP implies the existence of an instance optimal decoder for a certain norm we call the ``$M$-norm''. \revision{Using a topological result, we illustrate that this implication may be exploited for certain models and norms, even in infinite dimensions.} We propose an upper-bound on this norm under a generalized upper-RIP assumption to get an IOP with simpler norms, illustrating the result in standard cases.

\section{Generalized IOP and NSP equivalences}
\label{sec:gen_io_nsp}

In this section, we review the initial IOP/NSP relationship before extending it in several ways.

\subsection{The initial result}

In \cite{Cohen09compressedsensing}, the authors consider two norms $\nx{.}$ and $\ny{.}$ defined on the signal space $\rn$. The distance derived from $\ny{.}$ will be denoted $d_E$. Given a vector $\bx\in\rn$ and a subset $A\subset\rn$, the distance from $\bx$ to $A$ is defined as $d_E(\bx,A)=\inf{\by\in A}{\ny{\bx-\by}}$. These two norms allow the definition of instance optimality: a decoder $\Delta:\rm\rightarrow\rn$ is said to be instance optimal for $k$-sparse signals if
\begin{equation}
\label{io_initial}
\forall\bx\in\rn, \nx{\bx-\Delta(\bM\bx)}\leq Cd_E(\bx,\Sigma_k),
\end{equation}
for some constant $C>0$.

This property on $\Delta$ upper bounds the reconstruction error of a vector, measured by $\nx{.}$, by the distance from the vector to the model, measured by $d_E$. The authors prove that the existence of an instance optimal decoder, called IOP, is closely related to the NSP of $\bM$ with respect to the set $\Sigma_{2k}$ of $2k$-sparse vectors. Noting $\N=\ker(\bM)$, this NSP states
\begin{equation}
\label{nsp_initial}
\forall\bh\in\N, \nx{\bh}\leq Dd_E(\bh,\Sigma_{2k})
\end{equation}
for some constant $D$.

The relationship between the IOP and the NSP is the following: if there exists an instance optimal decoder $\Delta$ satisfying \eqref{io_initial}, then \eqref{nsp_initial} holds with $D=C$. Conversely, if \eqref{nsp_initial} holds, then there exists a decoder $\Delta$ such that \eqref{io_initial} holds with $C=2D$. Such a decoder can be defined as follows, supposing $\bM$ is onto:
\begin{equation}
\label{decoder_initial}
\Delta(\bM\bx)=\argmin{\bz\in(\bx+\N)}{d_E(\bz,\Sigma_k)},
\end{equation}
$\bx+\N$ denoting the set $\{\bx+\bh,\bh\in\N\}$. The well-posedness of this definition is discussed in Appendix \ref{sec:sparse_decoder}, in the more general setting where the model is a finite union of subspaces in finite dimension. Note that for generalized models, such a decoder may not necessarily exist since the infimum of $d_E(\bz,\Sigma)$ may not be achieved, as we will discuss in the next section.

This result can be seen as an ``equivalence'' between the IOP and the NSP, with similar constants.
% As we will see, much of the value of IO is in the size of the constant.

\subsection{Proposed extensions}

The framework we consider is more general. The signal space is a vector space $E$, possibly infinite-dimensional. In particular, $E$ may be a Banach space such as an $L^p$ space \revision{or a space of signed measures}. On this space is defined a linear operator $\bM: E\rightarrow F$, where $F$ is the measurement space, which will most likely be finite-dimensional in practice. We assume that $\bM$ is onto. We further define a signal model $\Sigma\subset E$ comprising the signals which we want to be able to ``reconstruct'' from their images by $\bM$.
In the framework we consider, this ``reconstruction'' is not necessarily an inverse problem where we want to recover $\bx$ from $\bM\bx$. More precisely, as in \cite{PelegGD13}, we consider a case where we want to recover from $\bM\bx$ a quantity $\bA\bx$, where $\bA$ is a linear operator mapping $E$ into a space $G$. When $G=E$ and $\bA=\bI$, we are brought back to the usual case where we want to reconstruct $\bx$. This generalized framework is illustrated in Figure \ref{fig:generalized_model}.

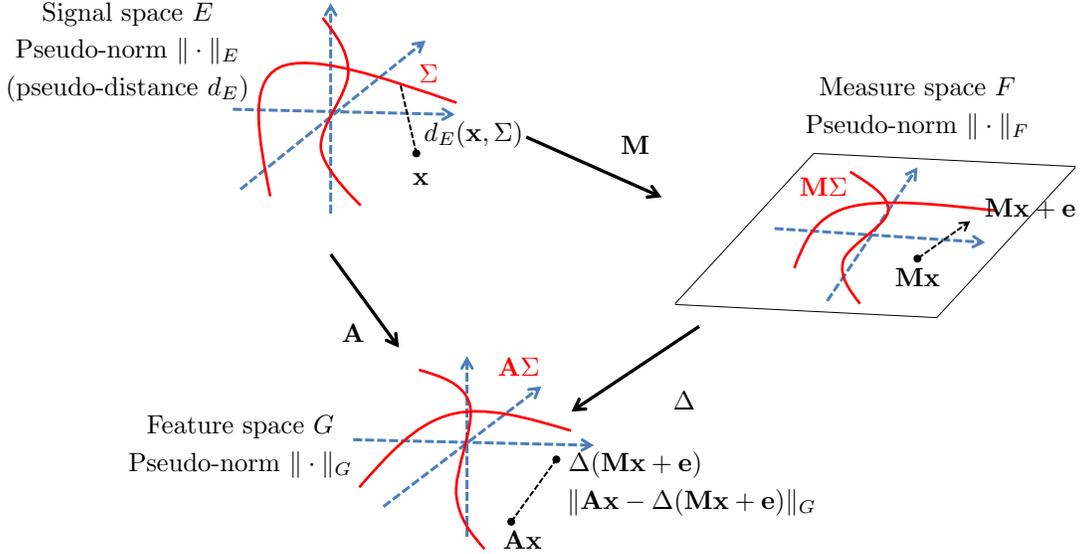
\begin{figure}
\centering
\begin{tikzpicture}
\draw (0,0) node{\includegraphics[scale=.5]{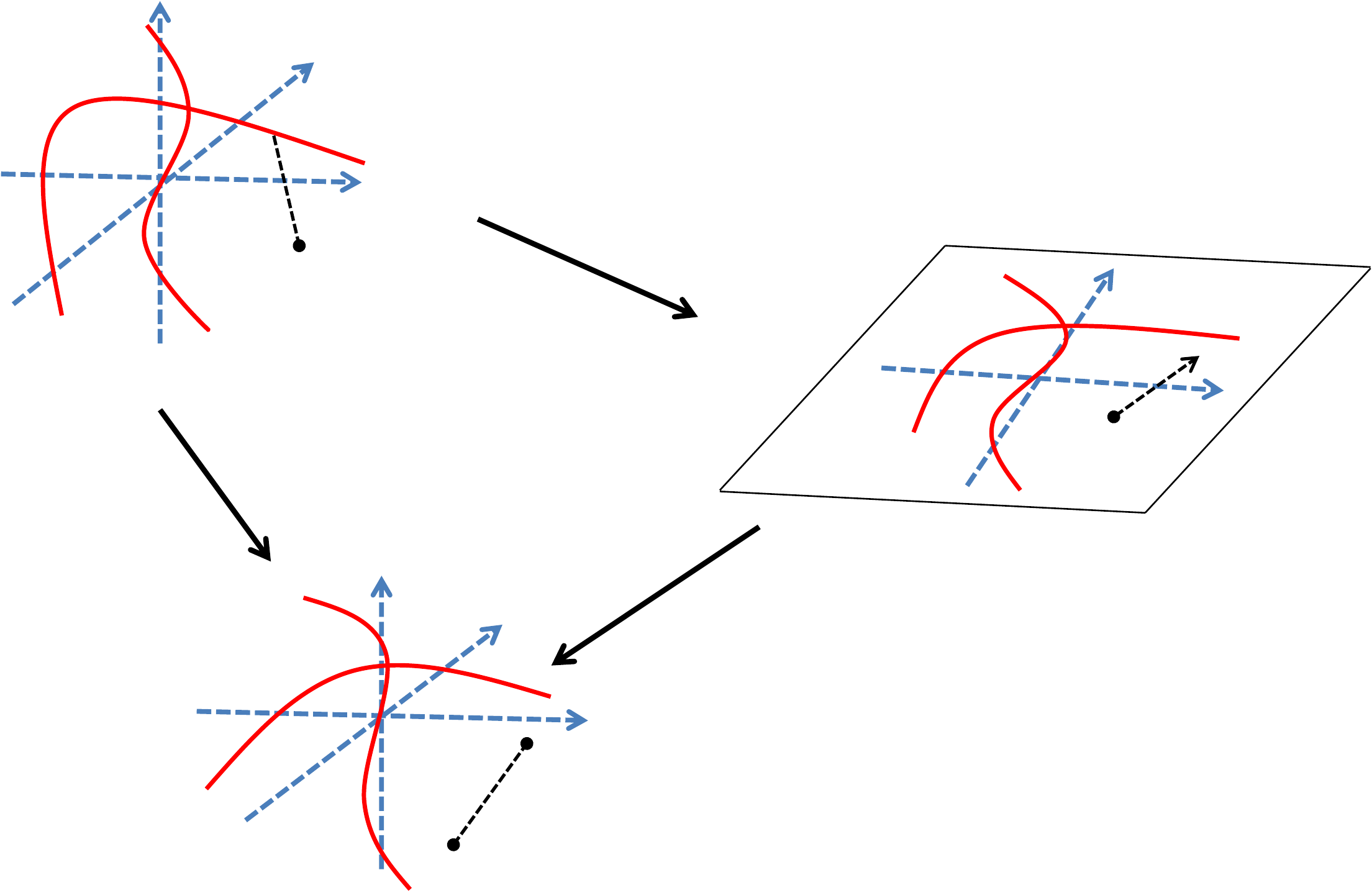}};
\draw (-7,3.5) node{Signal space $E$};
\draw (-7,3) node{Pseudo-norm $\ny{\cdot}$};
\draw (-7,2.5) node{(pseudo-distance $d_E$)};
\draw (3.5,2.5) node{Measure space $F$};
\draw (3.5,2) node{Pseudo-norm $\nz{\cdot}$};
\draw (-3,2.7) node{{\color{red} $\Sigma$}};
\draw (-3.1,1.3) node{$\bx$};
\draw (-2.4,1.9) node{$d_E(\bx,\Sigma)$};
\draw (3.5,0) node{$\bM\bx$};
\draw (5,.9) node{$\bM\bx+\be$};
\draw (2.25,1.2) node{{\color{red} $\bM\Sigma$}};
\draw (-1.75,-3.5) node{$\bA\bx$};
\draw (-.25,-2.5) node{$\Delta(\bM\bx+\be)$};
\draw (.5,-3) node{$\nx{\bA\bx-\Delta(\bM\bx+\be)}$};
\draw (-1.8,-1.2) node{{\color{red} $\bA\Sigma$}};
\draw (-0.25,1.75) node{$\bM$};
\draw (.4,-1.65) node{$\Delta$};
\draw (-4,-.75) node{$\bA$};
\draw (-5.5,-2) node{Feature space $G$};
\draw (-5.5,-2.5) node{Pseudo-norm $\nx{\cdot}$};
\end{tikzpicture}
\caption{Illustration of the proposed generalized setting. The signals belong to the space $E$, supplied with a pseudo-norm $\ny{.}$ used to measure the distance from a vector to the model $\Sigma$ containing the signals of interest. $E$ is mapped in the measure space $F$ by the operator $\bM$ and the measure is perturbed by an additive noise $\be$. The space $F$ is supplied with a pseudo-norm $\nz{.}$. The feature space $G$, supplied with a norm $\nx{.}$, is composed of vectors obtained by applying a linear operator $\bA$ to the signals in $E$. These feature vectors are the vectors one wants to reconstruct from the measures in $\bM$ by applying a decoder $\Delta$. The reconstruction error for the vector $\bx$ and noise $\be$ is therefore $\nx{\bA\bx-\Delta(\bM\bx+\be)}$. Note that in the case where $E=G$ and $\bA=\bI$, the decoder is aimed at reconstructing exactly the signals.}
\label{fig:generalized_model}
\end{figure}

%\subsubsection{Error measures}
In this generalized framework, we are now interested in the concepts of IOP and NSP, as well as their relationship. A decoder $\Delta: F\rightarrow G$ will aim at approximating $\bA\bx$ from $\bM\bx$.

The approximation error will be measured by a function $\nx{.}: G\rightarrow \r_+$. This function needs not be a norm in order to state the following results. It still must satisfy the following properties:
\begin{align}
\mathbf{Symmetry:}\; & \nx{\bx}=\nx{-\bx} \label{nx_symmetry}\\
\mathbf{Triangle\; inequality:}\; & \nx{\bx+\by}\leq \nx{\bx}+\nx{\by}. \label{nx_triangle}
\end{align}
The differences with a regular norm is that neither definiteness nor homogeneity is required: $\nx{\bx}=0$ needs not imply $\bx=0$ and $\nx{\lambda\bx}$ needs not equal $|\lambda|\nx{\bx}$. We provide two examples of such pseudo-norms in the case where $G=\rn$:
\begin{itemize}
\item $\nx{.}$ can be defined as a ``non-normalized'' $\ell^p$-quasinorm for $0\leq p\leq 1$, that is $\nx{\bx}=\sum_{i=1}^n |x_i|^p$. In this case, $\nx{\lambda\bx}=|\lambda|^p\nx{\bx}$.
\item More generally, if $f:\r_+\rightarrow\r_+$ is a concave function such that $f(x)=0\Leftrightarrow x=0$, then $\nx{.}$ can be defined as the $f$-(pseudo-)norm $\|\bx\|_f=\sum_{i=1}^n f(|x_i|)$, see \cite{gribonval07:_highl}.
\end{itemize}
In order to measure the \emph{distance from a vector to the model}, we also endow $E$ with a pseudo-norm $\ny{.}: E\rightarrow \r_+$ which satisfies the same properties as $\nx{.}$ with the additional requirement that $\ny{0}=0$. The pseudo-distance $d_E$ is defined on $E^2$ by $d(\bx,\by)=\ny{\bx-\by}$. Yet again, $\ny{.}$ can be defined as a non-normalized $\ell^p$-norm or an $f$-norm.

We will also consider a noisy framework where the measure $\bM\bx$ is perturbed by an additive noise term $\be$. To consider IOP and NSP in this context, we \emph{measure the amount of noise} with a pseudo-norm in the measurement space $F$, which we will denote by $\nz{.}$. The assumptions we make on $\nz{.}$ are the same as the assumptions on $\ny{.}$.

To sum up, here are the extensions we propose compared to the framework of \cite{Cohen09compressedsensing,PelegGD13} :
\begin{itemize}
\item The measure $\bM\bx$ can be perturbed by an additive noise $\be$.
\item The model set $\Sigma$ can be any subset of $E$.
\item $E$ is not necessarily $\rn$ but can be any vector space, possibly infinite-dimensional.
\item The reconstruction of $\bA\bx$ is targeted rather than that of $\bx$.
\item The functions $\ny{.}$, $\nz{.}$ and $\nx{.}$ need not be norms but can be pseudo-norms with relaxed hypotheses. In particular, Table \ref{tab:pseudo_norms} summarizes the requirements on these functions.
\end{itemize}

\begin{table}
\centering
\begin{tabular}{c|c|c|c|c|c}
 & Triangle Inequality & Symmetry & $\|0\|=0$ & Definiteness & Homogeneity\\
\hline
$\ny{.}$ & X & X & X & - & - \\
\hline
$\nz{.}$ & X & X & X & - & - \\
\hline
$\nx{.}$ & X & X & - & - & - 
\end{tabular}
\caption{Summary of the hypotheses on the pseudo-norms $\ny{.}$, $\nz{.}$ and $\nx{.}$. A cross means the property is required, a horizontal bar means it is not.}
\label{tab:pseudo_norms}
\end{table}

\revision{Once we have derived the generalized IOP and NSP equivalences, we see that one can essentially be brought back to the case where $\bA=\bI$ with a proper choice of $\gnorm{\cdot}$. This will be discussed in Section \ref{sec:task_io}.
}

\revision{
Let's note that even though \cite{Cohen09compressedsensing} does not consider the noisy case, some other works have studied noisy instance optimality for the standard sparse model and with $\ell^p$-norms (\cite{Wojtaszczyk10}, Chapter 11 of \cite{FoucartR13}). They mainly study conditions under which standard $\ell^1$ decoders are instance optimal. Here, we adopt a more conceptual approach by considering conditions for the \emph{existence} of an instance optimal decoder, without restriction on its practical tractability. This has the advantage of providing fairly simple equivalences and also to identify fundamental performance limits in a certain framework.

In these works, the underlined relationships between $\ell^1$ instance optimality and NSP are somewhat different than ours since they usually take advantage of the particular geometry of the sparse problem with $\ell^1$ decoder. An interesting open question is to what extent we can bridge the gap between this particular setup and a more general setup.
}

\subsubsection{The noiseless case}

We first consider the same framework as \cite{Cohen09compressedsensing,PelegGD13}, where one measures $\bM\bx$ with infinite precision. In our generalized framework, instance optimality for a decoder $\Delta$ reads:
%\begin{equation}
%\label{io_generalized}
\[
\forall\bx\in E,\nx{\bA\bx-\Delta(\bM\bx)}\leq C d_E(\bx,\Sigma).
\]
%\end{equation}
We will prove that if IOP holds, \ie, if the above holds for a certain decoder $\Delta$, then a generalized NSP is satisfied, that is:
%\begin{equation}
%\label{nsp_generalized}
\[
\forall\bh\in\N, \nx{\bA\bh}\leq Dd_E(\bh,\sms),
\]
%\end{equation}
with $D=C$. Note that the set $\Sigma_{2k}$ has been replaced by $\sms=\{\bx-\by,\bx\in\Sigma,\by\in\Sigma\}$. When $\Sigma=\Sigma_k$, we have indeed $\sms=\Sigma_{2k}$.

The construction of an instance optimal decoder from the NSP is more complicated and the form of the instance optimality we get depends on additional assumptions on $\Sigma$ and $\bM$. Let's first suppose that for all $\bx\in E$, there exists $\bz\in (\bx+\N)$ such that $d_E(\bz,\Sigma)=d_E(\bx+\N,\Sigma)$. Then the NSP \eqref{nsp_generalized} implies the existence of an instance optimal decoder satisfying \eqref{io_generalized} with $C=2D$. If this assumption is not true anymore, then the NSP implies a slightly modified IOP, which states, for any $\delta>0$, the existence of a decoder $\dd$ such that:
\begin{equation}
\label{io_generalized_nonexact}
\forall\bx\in E,\nx{\bA\bx-\dd(\bM\bx)}\leq C d_E(\bx,\Sigma)+\delta,
\end{equation}
reflecting the fact that one cannot necessarily consider the exact quantity 
\[
\argmin{\bz\in(\bx+\N)}{d_E(\bz,\Sigma)}
\]
but rather a certain vector $\bz\in(\bx+\N)$ satisfying $d_E(\bz,\Sigma)\leq d_E(\bx+\N,\Sigma) +\delta$. A similar positive ``projection error'' appears in \cite{Blumensath11}. 

\begin{remark}
To understand the necessity of such an additive error term when $\Sigma$ is a general set, we can consider the following toy example depicted in Figure \ref{fig:delta} where $E=\r^2$, $\N=\r\times\{0\}$,  $\Sigma=\{(x_1,x_2)\in(\r_+)^2:x_2=\frac{1}{x_1}\}$ and $\nx{.}/\ny{.}$ are the $\ell^2$ norm. In this case, the minimal distance between $\bx+\N$ and $\Sigma$ is not reached at any point, making it necessary to add the $\delta$ term for the decoder to be well-defined.
\end{remark}

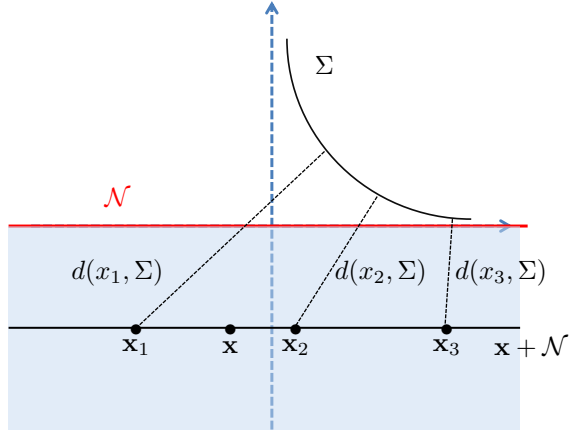
\begin{figure}
\centering
\begin{tikzpicture}
\draw (0,0) node{\includegraphics[scale=.5]{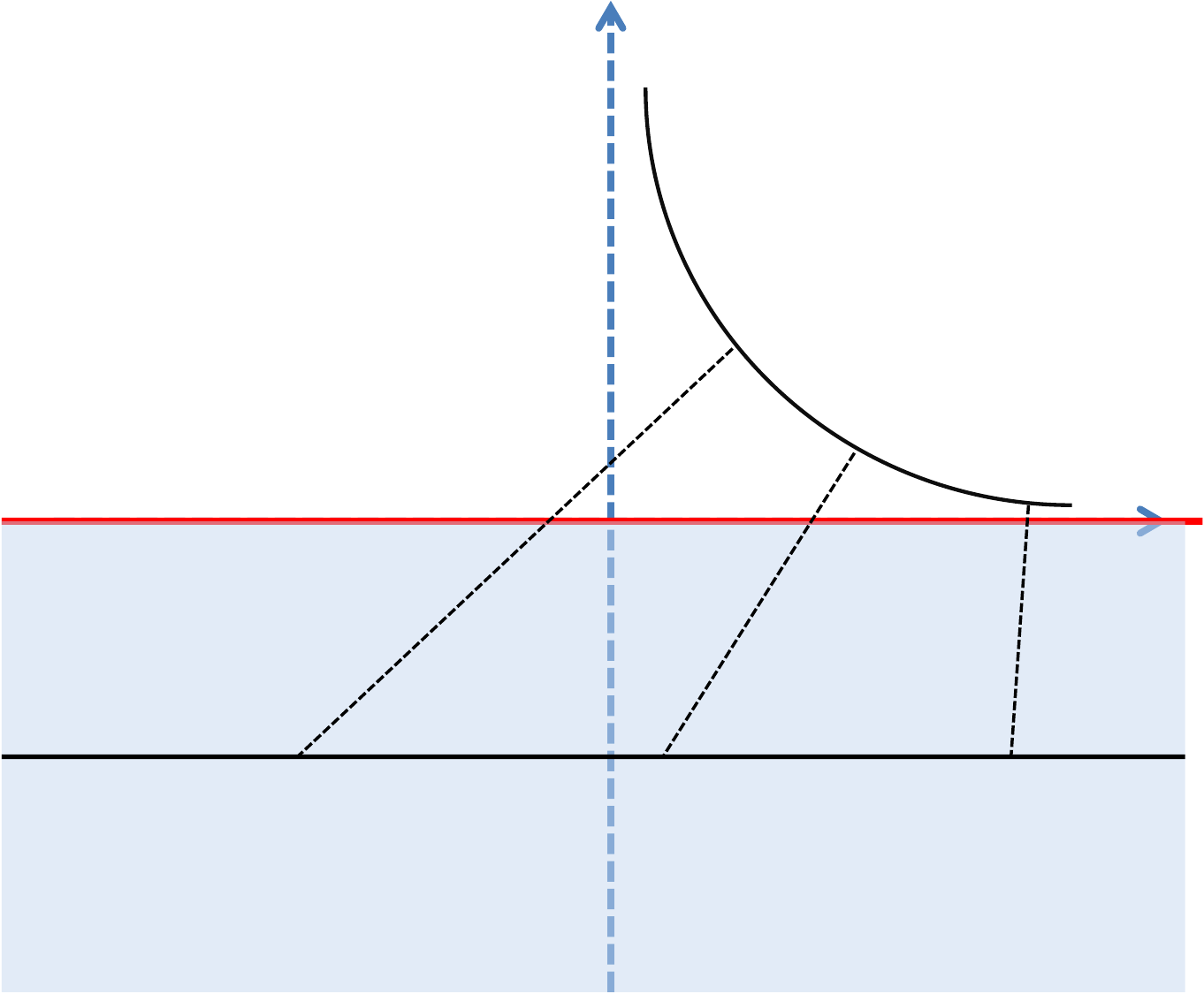}};
\draw (-2,.25) node{{\color{red} $\N$}};
\draw (.75,2) node{$\Sigma$};
\draw (3.5,-1.75) node{$\bx+\N$};
\draw (-.5,-1.52) node{\textbullet};
\draw (-.5,-1.75) node{$\bx$};
\draw (-1.75,-1.52) node{\textbullet};
\draw (-1.75,-1.75) node{$\bx_1$};
\draw (-2,-.75) node{$d(x_1,\Sigma)$};
\draw (.37,-1.52) node{\textbullet};
\draw (.37,-1.75) node{$\bx_2$};
\draw (1.5,-.75) node{$d(x_2,\Sigma)$};
\draw (2.38,-1.52) node{\textbullet};
\draw (2.38,-1.75) node{$\bx_3$};
\draw (3.1,-.75) node{$d(x_3,\Sigma)$};
\end{tikzpicture}

\caption{Necessity of the additive term $\delta$ in a simple case. For each $\bx$ in the blue half-plane, the distance $d_E(x+\N,\Sigma)$ is never reached at a particular point of $\bx+\N$ : the distance strictly decreases as one goes right along the affine plane $\bx+\N$ $(d(\bx_1,\Sigma)<d(\bx_2,\Sigma)<d(\bx_3,\Sigma))$, so that the minimal distance is reached ``at infinity''.}
\label{fig:delta}
\end{figure}

In this setting, the NSP \eqref{nsp_generalized} implies the existence of instance optimal decoders in the sense of \eqref{io_generalized_nonexact} for all $\delta>0$. Moreover, this weak IOP formulation still implies the regular NSP with $D=C$. This is summarized in Theorems \ref{thm:io_nsp} and \ref{thm:nsp_io}.

\begin{theorem}
\label{thm:io_nsp}
Suppose $\forall\delta>0$, there exists a decoder $\dd$ satisfying \eqref{io_generalized_nonexact}:
\begin{equation*}
\forall\bx\in E,\nx{\bA\bx-\dd(\bM\bx)}\leq C d_E(\bx,\Sigma)+\delta.
\end{equation*}
Then $\bM$ satisfies the NSP \eqref{nsp_generalized}:
%\begin{equation}
\[
\forall\bh\in\N, \nx{\bA\bh}\leq Dd_E(\bh,\sms),
\]
%\end{equation}
with constant $D=C$.
\end{theorem}

\begin{theorem}
\label{thm:nsp_io}
Suppose that $\bM$ satisfies the NSP \eqref{nsp_generalized}:
%\begin{equation}
\[
\forall\bh\in\N, \nx{\bA\bh}\leq Dd_E(\bh,\sms).
\]
%\end{equation}
Then $\forall\delta >0$, there exists a decoder $\dd$ satisfying \eqref{io_generalized_nonexact}:
%\begin{equation}
\[
\forall\bx\in E,\nx{\bA\bx-\dd(\bM\bx)}\leq C d_E(\bx,\Sigma)+\delta,
\]
%\end{equation}
with $C=2D$.

If we further assume that 
\begin{equation}
\label{exact_decoder}
\forall\bx\in E, \exists\bz\in (\bx+\N), d_E(\bz,\Sigma)=d_E(\bx+\N,\Sigma),
\end{equation}
then there exists a decoder $\Delta$ satisfying \eqref{io_generalized}:
\begin{equation}
\forall\bx\in E,\nx{\bA\bx-\Delta(\bM\bx)}\leq C d_E(\bx,\Sigma)
\end{equation}
with $C=2D$.
\end{theorem}

Note that this result is similar to the result proven in \cite{PelegGD13}, which was stated in the case where $\Sigma$ is a finite union of subspaces in finite dimension. In this framework, condition \eqref{exact_decoder} is always satisfied as soon as $\ny{.}$ is a norm, by the same argument as in usual CS (see Appendix \ref{sec:sparse_decoder}).

Let's also note the following property: if $\ny{.}$ is definite, that is $\ny{\bx}=0\Rightarrow \bx=0$, then $d_E$ is a distance. In the following proposition, we prove that if we further suppose that the set $\Sigma+\N$ is a closed set with respect to $d_E$, then the NSP \eqref{nsp_generalized} implies for any $\delta>0$ the existence of a decoder $\dd$ satisfying \eqref{io_generalized} with $C=(2+\delta)D$. This assumption therefore allows us to suppress the additive constant in \eqref{io_generalized_nonexact} and replace it by an arbitrarily small increase in the multiplicative constant of \eqref{io_generalized}.

\begin{proposition}
\label{prop:sigma_n_closed}
Suppose that $\bM$ satisfies the NSP \eqref{nsp_generalized}, that $d_E$ is a distance and that $\Sigma+\N$ is a closed set with respect to $d_E$. Then $\forall\delta >0$, there exists a decoder $\Delta_{\delta}$ satisfying:
\begin{equation}
\forall\bx\in E, \nx{\bA\bx-\Delta_{\delta}(\bM\bx)}\leq (2+\delta)D d_E(\bx,\Sigma).
\end{equation}
\end{proposition}

\subsubsection{The noisy case}

In practice, it is not likely that one can measure with infinite precision the quantity $\bM\bx$. This measure is likely to be contaminated with some noise, which will be considered in the following as an additive term $\be\in F$, so that the measure one gets is $\by=\bM\bx+\be$. In this case, a good decoder should be robust to noise, so that moderate values of $\be$ should not have a severe impact on the approximation error. We are interested in the existence of similar results as before in this noisy setting.

We first need to define a noise-robust version of instance optimality. 
The robustness to noise of practical decoders is in fact a problem that has been considered by many authors. 
A first type of result considers  {\em noise-aware decoders}, where given the noise level $\epsilon \geq 0$ a decoder $\d$ fulfills the following property
\begin{equation}
\label{io_gen_robust_known}
\forall\bx\in E, \forall\be\in F,\ \nz{\be} \leq \epsilon \Rightarrow \nx{\bA\bx-\d(\bM\bx+\be)}\leq C_1 d_E(\bx,\Sigma)+ C_{2} \epsilon.
\end{equation}
Here, the upper bound on the approximation error gets a new term measuring the amplitude of the noise. 
For example, this noise-robust instance optimality holds for a noise-aware $\ell^1$ decoder in the sparse case with bounded noise \cite{Candes08} for $\nx{.}=\ntwo{.}$ and $\ny{.}=\|.\|_1/\sqrt{k}$, provided $\bM$ satisfies the RIP on $\Sigma_{2k}$. 
%A similar property has been proven in \cite{Blumensath11} for a certain decoder under a similar assumption on $\bM$.

In practical settings, it is hard to assume that one knows precisely the noise level. To exploit the above guarantee with a noise-aware decoder, one typically needs to overestimate the noise level. This loosens the effective performance guarantee and potentially degrades the actual performance of the decoder.
An apparently stronger property for a decoder is to be robust {\em even without knowledge of the noise level}:
\begin{equation}
\label{io_gen_robust_unknown}
\forall\bx\in E, \forall\be\in F,\ \nx{\bA\bx-\d(\bM\bx+\be)}\leq C_1 d_E(\bx,\Sigma)+ C_{2} \nz{\be}.
\end{equation}
Further on, such decoders will be referred to as {\em noise-blind}. Guarantees of this type have been obtained under a RIP assumption for practical decoders such as iterative hard thresholding, CoSAMP, or hard thresholding pursuit, see e.g. \cite[Corollary 3.9]{Fou11}.

Of course, the existence of a {\em noise-blind} noise-robust decoder in the sense of~\eqref{io_gen_robust_unknown} implies the existence of a {\em noise-aware} noise-robust decoder in the sense of~\eqref{io_gen_robust_known} for any noise level $\epsilon$. We will see that, somewhat surprisingly, the converse is true in a sense, for both are equivalent to a noise-robust NSP. %, under the mild assumption that the (pseudo)norms $N_{X}$, $N_{Y}$ and $N_{Z}$ are homogeneous with the same degree.

Just as in the noiseless case, dealing with an arbitrary model $\Sigma$ and possibly infinite dimensional $E$ requires some caution. For $\delta>0$, the noise-robust (and noise-blind) instance optimality of a decoder $\dd$ is defined as:
\begin{equation}
\label{io_gen_robust}
\forall\bx\in E, \forall\be\in F, \nx{\bA\bx-\dd(\bM\bx+\be)}\leq C_1 d_E(\bx,\Sigma)+ C_2\nz{\be}+\delta.
\end{equation}
One can see that $\dd$ necessarily also satisfies the noiseless instance optimality \eqref{io_generalized_nonexact} by setting $\be=0$.

As we show below, if for every $\delta>0$ there exists a noise-robust instance optimal decoder $\Delta_{\delta}$ satisfying~\eqref{io_gen_robust}, then a generalized NSP for $\bM$ relatively to $\sms$, referred to as Robust NSP, must hold:
\begin{equation}
\label{robust_nsp}
\forall\bh\in E, \nx{\bA\bh}\leq D_1d_E(\bh,\sms)+D_2\nz{\bM\bh},
\end{equation}
with $D_1=C_1$ and $D_2=C_2$. This property appears e.g. in \cite{FoucartR13} (Chap. 4) with $\nx{.}=\ny{.}=\none{.}$ and $\nz{.}$ any norm. Note that this Robust NSP concerns every vector of $E$ and not just the vectors of the null space $\N=\ker(\bM)$\footnote{In fact, unlike the NSP~\eqref{nsp_generalized},~\eqref{robust_nsp} is not purely a property of the null space $\N$ even though it implies the NSP. The name Robust NSP is thus somewhat improper, but has become a standard for this type of property.}. In the case where $\bh\in\N$, one retrieves the regular NSP. For other vectors $\bh$, another additive term, measuring the ``size'' of $\bM\bh$, appears in the upper bound. 

Conversely, the Robust NSP implies the existence of noise-robust instance optimal decoders $\dd$ satisfying \eqref{io_gen_robust} with $C_1=2D_1$ and $C_2=2D_2$ for all $\delta>0$. These results are summarized in Theorems \ref{thm:io_nsp_noise} and \ref{thm:nsp_io_noise}.

\begin{theorem}
\label{thm:io_nsp_noise}
Suppose $\forall\delta>0$, there exists a decoder $\d_{\delta}$ satisfying \eqref{io_gen_robust}:
%\begin{equation}
\[
\forall\bx\in E,\forall\be\in F,\ \nx{\bA\bx-\d_{\delta}(\bM\bx+\be)}\leq C_1 d_E(\bx,\Sigma)+ C_2 \nz{\be}+\delta.
\]
%\end{equation}
Then $\bM$ satisfies the Robust NSP \eqref{robust_nsp}:
%\begin{equation}
\[
\forall\bh\in E, \nx{\bA\bh}\leq D_1d_E(\bh,\sms)+D_2\nz{\bM\bh},
\]
%\end{equation}
with constants $D_1=C_1$ and $D_2=C_2$.
\end{theorem}

\begin{theorem}
\label{thm:nsp_io_noise}
Suppose that $\bM$ satisfies the Robust NSP \eqref{robust_nsp}:
%\begin{equation}
\[
\forall\bh\in E, \nx{\bA\bh}\leq D_1d_E(\bh,\sms)+D_2\nz{\bM\bh}.
\]
%\end{equation}
Then $\forall\delta >0$, there exists a decoder $\dd$ satisfying \eqref{io_gen_robust}:
%\begin{equation}
\[
\forall\bx\in E,\forall\be\in F, \nx{\bA\bx-\dd(\bM\bx+\be)}\leq C_1 d_E(\bx,\Sigma)+ C_2\nz{\be}+\delta,
\]
%\end{equation}
with constants $C_1=2D_1$ and $C_2=2D_2$.
\end{theorem}

%The robust NSP for $N_{X} = N_{Z} = \|\cdot\|_{2}$ and $N_{Y} = \|\cdot\|_{1}/\sqrt{k}$ appears implicitly in the proof of the main results of \cite{Candes08}, as a intermediate consequence of the RIP over $\Sigma_{2k}$ which is in turn used to show instance optimality for the $\ell^{1}$ decoder. The theorems above 
We conclude this section by discussing the relation between noise-aware and noise-blind decoders. A noise-aware version of noise-robust instance optimality can be defined where for $\epsilon \geq 0,\delta>0$ we require
\begin{equation}
\label{io_gen_robust_noise_aware}
\forall\bx\in E, \forall\be\in F,\ \nz{\be} \leq \epsilon \Rightarrow \nx{\bA\bx-\d_{\delta,\epsilon}(\bM\bx+\be)}\leq C_1 d_E(\bx,\Sigma)+ C_2 \epsilon+\delta.
\end{equation}
Of course, the existence of a noise-blind instance optimal decoder implies that of noise-aware decoders for every $\epsilon \geq 0$. The converse is indeed essentially true, up to the value of the constants $C_i$:
\begin{theorem}
\label{thm:io_nsp_noise_aware}
Suppose $\forall\epsilon,\delta>0$, there exists a {\em noise-aware} decoder $\d_{\delta,\epsilon}$ satisfying~\eqref{io_gen_robust_noise_aware}:
%\begin{equation}
\[
\forall\bx\in E, \forall\be\in F,\ \nz{\be} \leq \epsilon \Rightarrow \nx{\bA\bx-\d_{\delta,\epsilon}(\bM\bx+\be)}\leq C_1 d_E(\bx,\Sigma)+ C_2 \epsilon+\delta.
\]
%\end{equation}
Then $\bM$ satisfies the Robust NSP \eqref{robust_nsp} with constants $D_1=C_1$ and $D_2=2C_2$. Therefore, by Theorem~\ref{thm:nsp_io_noise}, there exists an instance optimal {\em noise-blind} decoder satisfying:
\[
\forall\bx\in E,\forall\be\in F, \nx{\bA\bx-\d_{\delta}(\bM\bx+\be)}\leq 2C_1 d_E(\bx,\Sigma)+ 4C_2 \nz{\be}+\delta.
\]
\end{theorem}

\revision{
\subsection{Task-oriented instance optimality}
\label{sec:task_io}

In this section, we show that the generalized instance optimality as stated in \eqref{io_generalized_nonexact} is essentially equivalent to the same property with $\bA=\bI$ and a different choice for the pseudo-norm $\gnorm{.}$.

Indeed, let's consider that one aims at reconstructing a certain feature $\bA\bx$ from the measurements $\bM\bx$. If for any $\delta>0$ there exists an instance optimal decoder $\dd$ such that \eqref{io_generalized_nonexact} is satisfied, then Theorem \ref{thm:io_nsp} ensures that NSP \eqref{nsp_generalized} is satisfied. Let's define the following pseudo-norm for any signal $\bx\in E$:
\begin{equation}
\norm{\bx}=\gnorm{\bA\bx}.
\end{equation}
The following NSP is satisfied:
\begin{equation}
\forall \bh\in\ker(\bM),\ \norm{\bh}\leq C d_E(\bh,\sms).
\end{equation}
Therefore, Theorem \ref{thm:nsp_io} ensures that there exists decoders $\dd^{\prime}:F\rightarrow G$ instance optimal in the following sense:
\begin{equation}
\norm{\bx-\dd^{\prime}(\bM\bx)}\leq 2C d_E(\bx,\Sigma)+\delta.
\end{equation}

This means that if a family of decoders $\dd$ aimed at decoding a feature $\bA\bx$ is instance optimal for the pseudo-norm $\gnorm{.}$, then there exists a family of decoders $\dd^{\prime}$ aimed at decoding the \emph{signal} $\bx$ which is instance optimal for the pseudo-norm $\norm{.}$ with a similar constant (up to a factor 2). Conversely, if there exists a family of decoders $\dd^{\prime}$ aimed at decoding $\bx$ which is instance optimal for the pseudo-norm $\norm{.}$, then a simple rewriting of the IOP gives that the decoders $\dd = \bA\dd^{\prime}$ are instance optimal for the pseudo-norm $\gnorm{.}$.

Therefore, IOP with a task-oriented decoder is essentially equivalent to IOP with a standard decoder provided a suitable change in the pseudo-norm is performed. The same reasoning can be applied to deduce this equivalence for Robust IOP. As a consequence, we will only consider the case $\bA=\bI$ in the remainder of the paper.
}

\section{$\ell^2/\ell^2$ Instance Optimality}
\label{sec:l2l2}
%From now on, we will consider the case where $\nx{.}$, $\ny{.}$ and $\nz{.}$ are norms with the same homogeneity constant, that is for all $\bx\in E$, $\nx{\lambda\bx}=|\lambda|^p\nx{\bx}$, and similarly with the same constant $p$ for $\ny{.}$ and $\nz{.}$. To simplify the proofs, we will suppose $p=1$ since it does not change the results. The norms will be denoted $\nx{.}$, $\ny{.}$ and $\nz{.}$. Moreover, we suppose $G=E$ and $\bA=\mathrm{Id}_E$ so that we want to recover $\bx$ from $\bM\bx$.
%

In this section, we suppose that $E$ is a Hilbert space equipped with the norm $\ntwo{.}$ and scalar product $\ps{.}{.}$, that $F=\rm$ and we consider a finite-dimensional subspace $V$ of dimension $n$, on which we define the measure operator $\bM$. We are interested in the following question in the noiseless framework: \emph{Is it possible to have a ``good'' noiseless instance optimal decoder with $\nx{.}=\ny{.}=\ntwo{.}$ in a dimensionality reducing context where $m\ll n$?}

A result of \cite{Cohen09compressedsensing} states that in the usual sparse setting, one cannot expect to get a good instance optimal decoder if $\bM$ performs a substantial dimensionality reduction, the best corresponding constant being $\sqrt{\frac{n}{m}}$. In \cite{PelegGD13}, the authors prove that this lower bound on the constant holds in the case where $\Sigma$ is a finite union of subspaces in finite dimension. Here, we are interested in a version of this result for the general case where $\Sigma$ can be a more general subset of $E$. More precisely, we will give a sufficient condition on $\Sigma$ under which the optimal $\ell^2/\ell^2$ instance optimalityconstant is of the order of $\sqrt{\frac{n}{m}}$, thus preventing the existence of a $\ell^2/\ell^2$ instance optimal decoder with small constant if $m\ll n$.

\subsection{Homogeneity of the NSP}

In the case where $\nx{.}$, $\ny{.}$ and $\nz{.}$ are homogeneous with the same degree, the general NSP can be rewritten as an NSP holding on the cone $\r(\sms)$ generated by $\sms$, i.e., the set $\{\lambda\bz | \lambda\in\r,\bz\in\sms\}$. 
\begin{lemma}
\label{lem:nsp_cone}

If $\nx{.}$ and $\ny{.}$ are homogeneous with the same degree, we have an equivalence between the NSP on $\sms$:
\begin{equation}
\label{nsp_sms}
\forall\bh\in\N, \nx{\bh}\leq Dd_E(\bh,\sms),
\end{equation}
and the NSP on $\r(\sms)$:
\begin{equation}
\label{nsp_cone}
\forall\bh\in\N, \nx{\bh}\leq Dd_E(\bh,\r(\sms)).
\end{equation}

Similarly, if $\nx{.}$, $\ny{.}$ and $\nz{.}$ are homogeneous with the same degree, we have an equivalence between the robust NSP on $\sms$:
\begin{equation}
\forall\bh\in E, \nx{\bh}\leq D_1d_E(\bh,\sms)+D_2\nz{\bM\bh},
\end{equation}
and the robust NSP on $\r(\sms)$:
\begin{equation}
\forall\bh\in E, \nx{\bh}\leq D_1d_E(\bh,\r(\sms))+D_2\nz{\bM\bh}.
\end{equation}
\end{lemma}
This lemma, which is valid even in the case where $\bA$ is not the identity, shows that the NSP imposes a constraint on the whole linear cone spanned by the elements of $\sms$ and not only on the elements themselves. Note that this equivalence is trivial in the case where $\Sigma$ is a union of subspaces since $\sms$ is already a cone in this case.

\subsection{The optimal $\ell^2/\ell^2$ NSP constant}

\label{sec:optimal_constant}

\begin{remark}
In the subsequent sections of the paper, we will assume that $\bA=\bI$ (this implies $G=E$), so that one aims at reconstructing the actual signal.
\end{remark}

In the $\ell^2/\ell^2$ case, one can give a simple definition of the optimal NSP constant $D_*$, that is the minimal real positive number $D$ such that the $\ell^2/\ell^2$ NSP is satisfied with constant $D$:
\begin{equation}
D_*=\inf\qquad \{D\in\r_+|\forall\bh\in\N, \ntwo{\bh}\leq Dd_2(\bh,\sms)\}.
\end{equation}
This definition assumes that there exists some constant so that the NSP is satisfied. Using the NSP definition and Lemma \ref{lem:nsp_cone}, we get that
\begin{equation}
D_* = \sup{\bh\in\N}{\sup{\bz\in\r(\sms)}{\frac{\ntwo{\bh}}{\ntwo{\bh-\bz}}}} = \sup{\bh\in\N\backslash\{0\}}{\sup{\bz\in\r(\sms)}{\frac{1}{\ntwo{\frac{\bh}{\ntwo{\bh}}-\frac{\bz}{\ntwo{\bh}}}}}}.
\end{equation}

Denoting $\Bt$ the unit ball for the $\ell^2$ norm, we can rewrite this last expression as:
\begin{equation}
D_* = \sup{\bh\in\N\cap\Bt}{\sup{\bz\in\r(\sms)}{\frac{1}{\ntwo{\bh-\bz}}}} = \sup{\bh\in\N\cap\Bt}{\sup{\bz\in\r(\sms)\cap\Bt}{\sup{\lambda\in\r}{\frac{1}{\ntwo{\bh-\lambda\bz}}}}}.
\end{equation}

A simple study gives that if $\ntwo{\bh}=\ntwo{\bz}=1$, then $\sup{\lambda\in\r}{\frac{1}{\ntwo{\bh-\lambda\bz}}}=\frac{1}{\sqrt{1-\ps{\bh}{\bz}^2}}$, so that:
\begin{equation}
\label{def_l2nsp_constant}
D_* = \sup{\bh\in\N\cap\Bt}{\sup{\bz\in\r(\sms)\cap\Bt}{\frac{1}{\sqrt{1-\ps{\bh}{\bz}^2}}}}.
\end{equation}

The contraposition of Theorem \ref{thm:io_nsp} gives the following result : if the NSP \eqref{nsp_generalized} is not satisfied for a certain constant $D$, then no decoder $\dd$ can satisfy instance optimality \eqref{io_generalized_nonexact} with constant $D$. In the $\ell^2/\ell^2$ case, considering $D<D_*$, $\bh\in\N\cap\Bt$ and $\bz\in\r(\sms)\cap\Bt$ such that $\ps{\bh}{\bz}^2\geq 1-\frac{1}{D^2}$, we can construct two vectors such that for any decoder, instance optimality with constant $<\sqrt{D^2-1}$ can only be satisfied for at most one of them. This will shed light on the link between NSP and IOP. We have $\bz=\frac{\bz_1-\bz_2}{\ntwo{\bz_1-\bz_2}}$ for some $\bz_1,\bz_2\in\Sigma$. Let $\Delta$ be a decoder. If $\Delta(\bM\bz_1)\ne\bz_1$, then this vector prevents $\Delta$ from being instance optimal. The same goes for $\bz_2$ if $\Delta(\bM\bz_2)\ne\bz_2$. Now, let's suppose that $\bz_1$ and $\bz_2$ are correctly decoded. In this case, $(\bz_1+\bz_2)/2$ is decoded with a constant worse than $\sqrt{D^2-1}$, as depicted in Figure \ref{fig:N_sigma_corr}. Indeed, noting $\bp=(\bz_1+\bz_2)/2$ and defining the vectors $\bp_1$ and $\bp_2$ respectively as the orthogonal projections of $\bz_1$ and $\bz_2$ on the affine plane $\bp+\N$, we must have $\Delta(\bM\bp_1)=\Delta(\bM\bp_2)$. Denoting as $p_{\No}$ the orthogonal projection on $\No$, we have $d_2(\bp_1,\Sigma)\leq d_2(\bp_1,\bz_1)=\ntwo{p_{\No}(\bz_2-\bz_1)}/2$. Similarly, $d_2(\bp_2,\Sigma)\leq \ntwo{p_{\No}(\bz_2-\bz_1)}/2$. The fact that $\Delta(\bM\bp_1)=\Delta(\bM\bp_2)$ implies that there exists $i\in\{1,2\}$ such that $\ntwo{\bp_i-\Delta(\bM\bp_i)}\geq \ntwo{\bp_1-\bp_2}/2=\ntwo{p_{\N}(\bz_1-\bz_2)}/2$. Therefore,
\begin{equation}
\frac{\ntwo{\bp_i-\Delta(\bM\bp_i)}}{d_2(\bp_i,\Sigma)}\geq \frac{\ntwo{p_{\N}(\bz_2-\bz_1)}}{\ntwo{p_{\No}(\bz_2-\bz_1)}}\geq D\sqrt{1-\frac{1}{D^2}}=\sqrt{D^2-1}.
\end{equation} 

This illustrates the closeness between NSP and IOP: a vector of $\r(\sms)$ which is correlated with $\N$ can be used to define a couple of vectors such that for any decoder, one of the vectors will not be well decoded.

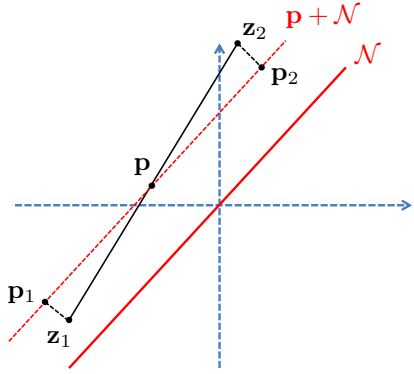
\begin{figure}

%\begin{minipage}{.49\linewidth}
%\begin{tikzpicture}
%\draw (0,0) node{\includegraphics[scale=.9]{sigma_corr_good.pdf}};
%\draw (1.8,2.6) node{{\color{red} $\N$}};
%\draw (2.8,.8) node{$\bx_1\in\Sigma$};
%\draw (-1,2.6) node{$\bx_2\in\Sigma$};
%\draw (-2.6,2.1) node{$\bx_1-\bx_2\in\sms$};
%\end{tikzpicture}
%\end{minipage}
%\begin{minipage}{.49\linewidth}
%\begin{tikzpicture}
%\draw (0,0) node{\includegraphics[scale=.9]{sigma_corr_bad.pdf}};
%\draw (1.8,2.6) node{{\color{red} $\N$}};
%\draw (-2.5,1.1) node{$\bx_1\in\Sigma$};
%\draw (1,3) node{$\bx_2\in\Sigma$};
%\draw (-3,-2.8) node{$\bx_1-\bx_2\in\sms$};
%\end{tikzpicture}
%\end{minipage}
\centering

\begin{tikzpicture}
\draw (0,0) node{\includegraphics[scale=.4]{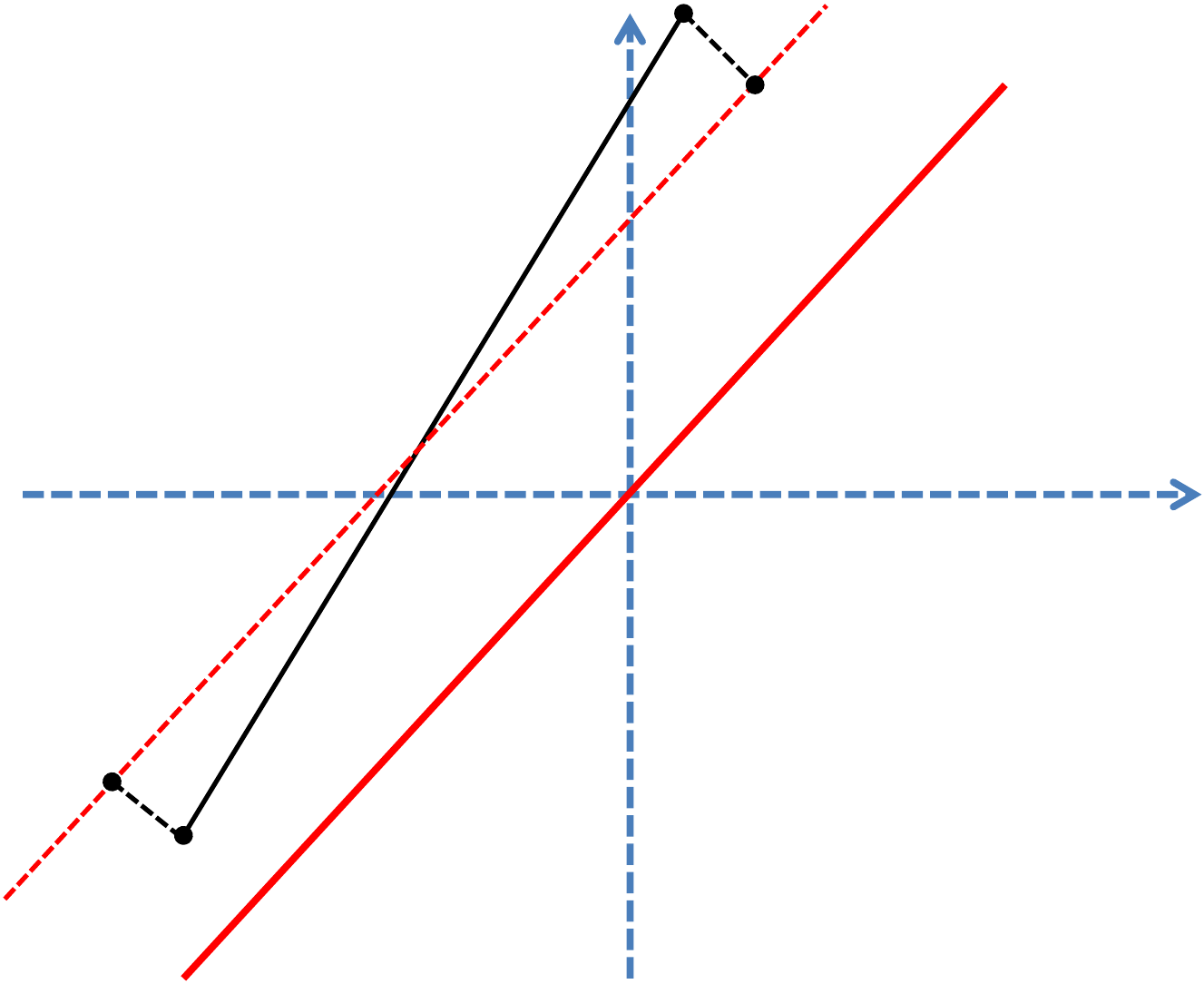}};
\draw (2.1,2) node{{\color{red} $\N$}};
\draw (-2.5,-1.2) node{$\bp_1$};
\draw (-2,-1.8) node{$\bz_1$};
\draw (.6,2.3) node{$\bz_2$};
\draw (1,1.7) node{$\bp_2$};
\draw (-.77,.25) node{{\tiny \textbullet}};
\draw (-.9,.5) node{$\bp$};
\draw (1.5,2.5) node{{\color{red} $\bp+\N$}};
\end{tikzpicture}

\caption{Illustration of the impact of the correlation between $\N$ and $\sms$ on instance optimality. Here, $\bz_1$ and $\bz_2$ are two vectors in $\Sigma$ such that $\bz_1-\bz_2$ is well correlated with $\N$, implying that at least one of the two vectors $\bp_1$ and $\bp_2$, which are close to $\Sigma$ but far from one another, will not be well decoded.}
\label{fig:N_sigma_corr}
\end{figure}

\subsection{$\ell^2/\ell^2$ IO with dimensionality reduction}

\subsubsection{Main theorem}

Let's now exploit the expression of $D_*$ to state the main result of this section: if $\r(\sms)$ contains an orthonormal basis of the finite-dimensional subspace $V\subset E$ (or even a family of vectors that is sufficiently correlated with every vector of $V$), then one cannot expect to get a $\ell^2/\ell^2$ instance optimal decoder with a small constant while $\bM$ substantially reduces the dimension of $V$. The fact that $\r(\sms)$ contains such a tight frame implies that the dimension of $\N$ cannot be too big without $\N$ being strongly correlated with $\sms$, thus yielding the impossibility of a good instance optimal decoder.

Before showing examples where this theorem applies, let's first state it and prove it.

\begin{theorem}
\label{thm:l2l2_io}
Suppose $V$ is of dimension $n$ and $\sms$ contains a family $\bz_1,\ldots,\bz_n$ of unit-norm vectors of $E$ satisfying $\forall \bx\in V$, $\sum_{i=1}^n \ps{\bz_i}{\bx}^2 \geq K\ntwo{\bx}^2$. Then to satisfy the NSP on $V$, $\bM$ must map $V$ into a space of dimension at least $\left(1-\frac{1}{K}\left(1-\frac{1}{D_*^2}\right)\right)n$.

If the number of measurements $m$ is fixed, then an $\ell^2/\ell^2$ IO decoder must have a constant at least $\frac{1}{\sqrt{1-K\left(1-\frac{m}{n}\right)}}$.
\end{theorem}

In particular, if $\sms$ contains an orthonormal basis of $V$, then $K=1$ and the minimal number of measures to achieve NSP with constant $D_*$ is $n/D_*^2$. Similarly, if $m$ is fixed so that $m\ll n$, then a $\ell^2/\ell^2$ instance optimal decoder has constant at least $\sqrt{\frac{n}{m}}$.

\subsubsection{Examples}

As discussed in the introduction, there is a wide range of standard models where $\sms$ contains an orthonormal basis, and so where $\ell^2/\ell^2$ IOP with dimensionality reduction is impossible. We provide here less trivial examples, where $E=V$ is finite-dimensional.

\paragraph{Symmetric definite positive matrices with sparse inverse.}
\label{sec:sparse_inverse}

\begin{lemma}
\label{lem:sdp_mat}
Consider $E$ is the space of symmetric $n$-dimensional matrices, and $\Sigma \subset E$ the subset of symmetric positive-definite matrices with sparse inverse and with sparsity constant $k\geq n+2$ (note that $k \geq n$ is necessary for the matrix to be invertible). The set $\sms$ contains an orthonormal basis of $E$. 
\end{lemma}
\begin{proof}
This orthonormal basis we consider is made of the $n(n+1)/2$ matrices: $\bE_{i,i}$ and $\frac{1}{\sqrt{2}}(\bE_{i,j}+\bE_{j,i})_{i\ne j}$, where $\bE_{i,j}$ is the matrix where the only nonzero entry is the $(i,j)$ entry which has value 1. 
%To express elements of this basis as differences between elements in $\Sigma$, we 

First, consider $\bB_i=\bI+\bE_{i,i}$, where $\bI$ is the identity matrix. Since $\bB_i^{-1}=\bI-\frac{1}{2}\bE_{i,i}$ is $n$-sparse, we have $\bB_{i} \in \Sigma$. Since, $\bI \in \Sigma$, we have $\bE_{i,i} = \bB_{i}-\bI \in \sms$.

Now, consider the matrix $\bC_{i,j}=2\bI+\bE_{i,j}+\bE_{j,i}$. This matrix is symmetric and for $\bx=(x_1,\ldots,x_n)\in\rn$, we have $\bx^T\bC_{i,j}\bx=2(\ntwo{\bx}^2-x_ix_j)\geq 0$, so that $\bC_{i,j}$ is semi-definite positive. We can remark that $\bC_{i,j}$ is invertible and that its inverse is $\frac12 \bI+\frac16 (\bE_{i,i}+\bE_{j,j})-\frac13 (\bE_{i,j}+\bE_{j,i})$, which is $n+2$-sparse. The fact that $\bC_{i,j}$ is invertible implies that it is definite, so that $\bC_{i,j}\in\Sigma$. Therefore, we can write $\bE_{i,i}+\bE_{j,j}=\bC_{i,j}-2\bI\in\sms$. Since $\Sigma$ is a positive cone, multiplying this equality by $\frac{1}{\sqrt{2}}$ yields the desired result.
%Therefore, $\sms$ contains an orthonormal basis of the set of $n$-dimensional symmetric matrices and one cannot expect good compressive $\ell^2/\ell^2$ IO recovery of such elements.
\end{proof}

\paragraph{Low-rank and nonsparse matrices.}
\label{sec:low_rank_nonsparse}

In \cite{CandesLMW11}, the authors consider a matrix decomposition of the form $\bL+\bS$, where $\bL$ is low-rank and $\bS$ is sparse. In order to give meaning to this decomposition, one must avoid $\bL$ to be sparse. To this end, a ``nonsparsity model'' for low-rank matrices was introduced.

Let $E$ be the space of complex matrices of size $n_1\times n_2$. Given $\mu\geq 1$ and $r\leq\min(n_1,n_2)$, let $\Sigma_{\mu,r}$ be the set of matrices of $E$ of rank $\leq r$ satisfying the two following conditions (denoting the SVD of such a matrix by $\sum_{k=1}^r \sigma_k\bu_k\bv_k^*$, where $\sigma_k>0$ and the $\bu_k$ and $\bv_k$ are unit-norm vectors) :
\begin{enumerate}
\item $\forall k, \ninf{\bu_k}\leq \sqrt{\frac{\mu r}{n_1}}$ and $\ninf{\bv_k}\leq \sqrt{\frac{\mu r}{n_2}}$.
\item Denoting $\bU$ and $\bV$ the matrices obtained by concatenating the vectors $\bu_k$ and $\bv_k$, $\ninf{\bU\bV^*}\leq \sqrt{\frac{\mu r}{n_1n_2}}$.
\end{enumerate}
These two conditions aim at ``homogenizing'' the entries of $\bU$ and $\bV$. Note that we necessarily have $\mu\geq 1$.

\begin{lemma}
\label{lem:lr_nonsparse_mat}
Let $E=\mathcal{M}_{n_1,n_2}(\mathbb{C})$ and $\Sigma_{\mu,r}$ be the subset of $E$ containing the matrices satisfying the two above conditions (with $\mu\geq 1$ and $r\geq 1$). Then $\Sigma_{\mu,r}-\Sigma_{\mu,r}$ contains an orthonormal basis.
\end{lemma}
\begin{proof}
Since $\Sigma_{\mu,r}$ contains the null matrix, it is sufficient to prove that $\Sigma_{\mu,r}$ contains an orthonormal basis. Let $\{\be_k\}_{k=1}^{n_1}$ and $\{\bf_{\ell}\}_{\ell=1}^{n_2}$ be the discrete Fourier bases of $\mathbb{C}^{n_1}$ and $\mathbb{C}^{n_2}$, that is
\begin{align*}
& \be_k=\frac{1}{\sqrt{n_1}}\left[1,\exp(2i\pi k/n_1),\ldots,\exp(2i\pi (n_1-1)k/n_1)\right]^T\\
\mathrm{and} & \;\bf_{\ell}=\frac{1}{\sqrt{n_2}}\left[ 1,\exp(2i\pi \ell/n_2),\ldots,\exp(2i\pi (n_2-1)\ell/n_2)\right]^T.
\end{align*}

Then the $n_1n_2$ rank-1 matrices of the form $\be_k\bf_{\ell}^*$ are elements of $\Sigma_{\mu,r}$ since they obviously satisfy the two above conditions. But they also form an orthonormal basis of $E$, since each entry of $\be_k\bf_{\ell}^*$ is of module $\frac{1}{\sqrt{n_1n_2}}$ and that, denoting $\ps{.}{.}$ the Hermitian scalar product on $E$,
\begin{equation}
\ps{\be_k\bf_{\ell}^*}{\be_{k^{\prime}}\bf_{\ell^{\prime}}^*}=\sum_{u=0}^{n_1-1}\exp\left(2i\pi u\frac{k-k^{\prime}}{n_1}\right)\sum_{v=0}^{n_2-1}\exp\left(2i\pi v\frac{\ell-\ell^{\prime}}{n_2}\right)=\delta_{k}^{k^{\prime}}\delta_{\ell}^{\ell^{\prime}},
\end{equation}
proving that these matrices form an orthonormal basis of $E$.
\end{proof}

%\subsubsection{Johnson-Lindenstrauss embedding}
%
%The point cloud embedding \cite{Ach01} constitutes a good example to interpret the meaning of the impossibility of a good $\ell^2/\ell^2$ IO with both small constant and substantial dimensionality reduction.
%
%In $\rn$, we consider a set $\X$ of $p$ points and want to find a linear operator $\bM:\rn\rightarrow\rm$ such that for all $\bx,\by\in\X$, $(1-\eps)\|\bx-\by\|\leq \|\bM\bx-\bM\by\| \leq (1+\eps)\|\bx-\by\|$ for a small $\eps>0$. This is feasible with high probability for several random choices of $\bM$ with $m=\mathcal{O}(\ln(p)/\eps^2)$. However, 
%
%\textbf{Interpretation of the impossibility of $\ell^2/\ell^2$ IO in this case.}

\section{The NSP and its relationship with the RIP}
\label{sec:nsp_rip}

\revision{As we have seen in the previous section,} one cannot expect to get $\ell^2/\ell^2$ instance optimality in a dimensionality reduction context. This raises the following question: given pseudo-norms $\nx{.}$ and $\nz{.}$ defined respectively on $E$ and $F$, is there a pseudo-norm $\ny{.}$ such that IOP holds? We will see that this property is closely related to the RIP on $\bM$.

\subsection{Generalized RIP and its necessity for robustness}

The Restricted Isometry Property is a widely-used property on the operator $\bM$ which yields nice stability and robustness results on the recovery of vectors from their compressive measurements. In the usual CS framework, the RIP provides a relation of the form $(1-\delta)\nx{\bx}\leq \nz{\bM\bx} \leq (1+\delta)\nx{\bx}$ for any vector $\bx$ in $\Sigma_{2k}$. The norms $\nx{.}$ and $\nz{.}$ are usually both taken as the $\ell^2$-norm. A form of RIP can easily be stated in a generalized framework: we will say that $\bM$ satisfies the RIP on $\sms$ if there exists positive constants $\alpha,\beta$ such that
\begin{equation}
\label{generalized_rip}
\forall\bz\in\sms, \alpha\nx{\bz}\leq \nz{\bM\bz} \leq \beta\nx{\bz}.
\end{equation}
Similarly to the sparse case, it is possible to make a distinction between \emph{lower-RIP} (left inequality) and \emph{upper-RIP} (right inequality). Let's remark that this definition has been stated for vectors of $\sms$: this choice is justified by the links between this formulation and the NSP, which will be discussed later in this section. Let's also note that this form of RIP encompasses several generalized RIP previously proposed, as mentioned in Section \ref{sec:gen_rip}.

Let's now suppose the existence of decoders robust to noise, that is for all $\delta>0$, \eqref{io_gen_robust} is satisfied for a certain $\dd$. This property implies the Robust NSP \eqref{robust_nsp} with the same constants according to Theorem \ref{thm:io_nsp_noise}. By considering $\bh\in\sms$, the Robust NSP reads:
\begin{equation}
\forall\bh\in\sms, \nx{\bh}\leq D_2\nz{\bM\bh}.
\end{equation}
This is the lower-RIP on $\sms$, with constant $1/D_2$. The stability to noise therefore implies the lower-RIP on the set of differences of vectors of $\Sigma$, which is therefore necessary if one seeks the existence of a decoder robust to noise.

\subsection{$M$-norm instance optimality with the RIP}

The lower-RIP is necessary for the existence of a Robust instance optimal decoder, but what can we say this time if we suppose that $\bM$ satisfies the lower-RIP on $\sms$ with constant $\alpha$, that is $\forall\bz\in\sms, \alpha\nx{\bz}\leq \nz{\bM\bz}$? We will prove that in both the noiseless and the noisy cases, this implies the IOP with norms $\nx{.}$ and $\nM{.}$, the latter being called ``$M$-norm''\footnote{to highlight its dependency on the Measurement operator} and involving $\nx{.}$ and $\nz{.}$.

Let's define the $M$-norm on $E$ as the following quantity, extending its definition for $\ell^2$ norms in \cite{PelegGD13} and its implicit appearance in the proof of early results of the field \cite{Candes08}:
\begin{equation}
\label{m_norm}
\forall\bx\in E, \nM{\bx}=\nx{\bx}+\frac{1}{\alpha}\nz{\bM\bx}.
\end{equation}

Note that the term $M$-norm should be understood as $M$-pseudo-norm in the general case: if $\nz{.}$ and $\nx{.}$ satisfy the properties listed in Table \ref{tab:pseudo_norms}, then $\nM{.}$ satisfies the same properties as $\nx{.}$. However, when $\nx{.}$ and $\nz{.}$ are norms, $\nM{.}$ is also a norm. We will note $\dM{.,.}$ its associated (pseudo-)distance. The following theorem states that this $\nM{.}$ allows one to derive an NSP from the lower-RIP on $\sms$.

\begin{theorem}
\label{thm:nsp_mnorm}
Let's suppose that $\bM$ satisfies the lower-RIP on $\sms$ with constant $\alpha$ (left inequality of \eqref{generalized_rip}). Then the following Robust NSP is satisfied:
\begin{equation}
\forall\bh\in E, \nx{\bh}\leq \dM{\bh,\sms}+\frac{1}{\alpha}\nz{\bM\bh}.
\end{equation}
In particular, the following regular NSP is satisfied:
\begin{equation}
\forall\bh\in\N, \nx{\bh}\leq \dM{\bh,\sms}.
\end{equation}
\end{theorem}

Therefore, if $\bM$ satisfies the lower-RIP on $\sms$ with constant $\alpha$, then for all $\delta>0$, there exists a noise-robust instance optimal decoder $\dd$ satisfying the following property (Theorem \ref{thm:nsp_io_noise}):
\begin{equation}
\label{eq:io_mnorm}
\forall\bx\in E, \forall\be\in F, \nx{\bx-\dd(\bM\bx+\be)}\leq 2\dM{\bx,\Sigma}+\frac{2}{\alpha}\nz{\be}+\delta.
\end{equation}

Note that in \cite{Blumensath11}, the author explored the implication of a lower-RIP on
$\sms$ for the case where $\Sigma$ is an arbitrary UoS and
$\|.\|_G$/$\|.\|_F$ are the $\ell^2$ norm. He proved that this
generalized lower-RIP implies the following IOP: for all $\delta>0$, there exists a decoder $\dd$ such that
\begin{equation}
\label{io_blumensath}
\forall\bx\in E, \forall\be\in F, \forall\bz\in\Sigma, \ntwo{\bx-\dd(\bM\bx+\be)}\leq \ntwo{\bx-\bz} + \frac{2}{\alpha}\ntwo{\bM(\bx-\bz)+\be} + \delta.
\end{equation}

In this set-up, the instance optimality in equation \eqref{eq:io_mnorm} can be reformulated as
\begin{equation}
\forall\bx\in E, \forall\be\in F, \forall\bz\in\Sigma, \ntwo{\bx-\dd(\bM\bx+\be)}\leq 2\ntwo{\bx-\bz} + \frac{2}{\alpha}\ntwo{\bM(\bx-\bz)}+\frac{2}{\alpha}\ntwo{\be} + \delta.
\end{equation}

Comparing these two instance optimality results, we can remark that the one in \cite{Blumensath11} is slightly tighter. This is merely a consequence of the difference in our method of proof, as we add the NSP as an intermediate result to prove instance optimality. The upper bound in \cite{Blumensath11} can also be derived in our case with the same proof layout if we suppose the lower-RIP. Compared to \cite{Blumensath11}, our theory deals with general (pseudo-)norms and sets $\Sigma$ beyond Union of Subspaces.

\revision{
\subsection{Infinite-dimensional examples}

As mentioned in the introduction, we do not constrain the signal space to be finite-dimensional, so that we can apply our results in an infinite-dimensional framework.

\subsubsection{Negative example: Sparse model in a separable Hilbert space}
\label{sec:sparse_hilbert}

If $E=L^2([0,1])$ and $\{\varphi_n\}_{n\in\n}$ is an orthonormal basis of $E$, a typical measurement process of a signal $\bx$ is to subsample along another orthonormal basis $\{\Psi_n\}_{n\in\n}$. Typically, $\{\varphi_n\}$ is a wavelet basis and $\{\Psi_n\}$ a Fourier basis. In \cite{BenAdcock:2013va}, the authors argue that standard sparsity does not represent well natural signals, which are rather \emph{asymptotically sparse}, that is more sparse at fine levels than at coarse levels. They propose an asymptotic sparsity model with different levels of sparsity on different scales.

Indeed, as the authors mention in their paper, a standard sparsity model $\Sigma$ with respect to basis $\{\varphi_n\}$ cannot yield uniform recovery for the $L^2$ norm: this is an obvious consequence of Theorem \ref{thm:l2l2_io} if $d_E$ is the $L^2$ distance and it is actually true for any other distance. Indeed, the NSP can never be satisfied since one has $\ntwo{\varphi_n+\varphi_{n+1}}=\sqrt{2}$ for all $n$ (since the family $\{\varphi_n\}$ is orthonormal) while the right hand side term of the NSP for $\bh_n=\varphi_n+\varphi_{n+1}$ is equal to
\begin{equation}
d(\bh_n,\sms)+\ntwo{\bM\bh_n}=\ntwo{\bM\bh_n},
\end{equation}
which goes to 0 when $n\rightarrow\infty$ (since $\bM$ is continuous).

\subsubsection{Positive example: Topological RIP result for $\Sigma$ of finite box-counting dimension}
\label{sec:topo_rip}

Even though the IOP cannot be satisfied for the standard sparse model in a Hilbert space, it does not mean IOP is impossible for all models in an infinite-dimensional space. Let's mention the following topological result, which is Theorem 8.1 in \cite{Robinson10} and ensures that a RIP is satisfied in some settings:

%Actually, an IOP has already been proven in such a context in Theorem 1.2 in \cite{rawa13}. In their case, the model $\Sigma$ is a standard sparsity model along a countable orthonormal family $\{\varphi_n\}_{n\in\n}$, $\ny{\cdot}$ is a weighted $\ell^1$ norm and $\nx{\cdot}$ is either the infinity norm or the $\ell^2$ norm.

\begin{theorem}
Let $\Sigma$ be a compact subset of a Banach space $\B$ supplied with norm $\norm{.}_{\B}$. Suppose $\Sigma$ has finite (upper) box-counting dimension $d$. Then for any $m>2d$, any norm $\norm{.}$ on $\rm$, and any $\theta$ satisfying
\[
0<\theta < \frac{m-2d}{m(1+d)},
\]
there exists a prevalent set\footnote{A prevalent set being a set which complementary is negligible in a certain sense. Definition is given in \cite{Robinson10}.} of continuous linear operators $\bM:\B\rightarrow \rm$ such that for any $\bx,\by\in\Sigma$,
\begin{equation}
C_{\bM}\norm{\bx-\by}_{\B}\leq \norm{\bM\bx-\bM\by}^{\theta}.
\end{equation}
\end{theorem}

This theorem essentially says that if $\Sigma$ has finite upper box-counting dimension\footnote{This is a notion of dimension defined by asymptotic behavior of $\epsilon$-covers.}, then a lower-RIP is satisfied for a prevalent set of operators $\bM$ with the pseudo-norms $\norm{.}_{\B}$ and $\norm{.}^{\theta}$ -- this last one being a pseudo-norm since $\theta\leq 1$. According to the previous section, this implies an IOP with the corresponding $M$-norm, that is the existence of a family of decoders $\dd$ such that for all $\bx\in E$, $\be\in F$ and $\bz\in\Sigma$:

\begin{equation}
\label{eq:iop_topological}
\norm{\bx-\Delta(\bM\bx)}_{\B}\leq 2d_M(\bx,\Sigma)+\frac{2}{C_{\bM}}\norm{\be}^{\theta}+\delta.
\end{equation}

We necessarily have $\theta\leq \frac{1}{d}$, meaning the exponent drops to 0 as $d$ grows, and therefore the corresponding IOP becomes much less powerful with a high-dimensional set $\Sigma$ (in the sense of the upper box-counting dimension). However, this essentially proves that uniform instance optimality is possible even in infinite dimensions with appropriate $\Sigma$ and pseudo-norms. Furthermore, weakening the existence of a \emph{prevalent} set of operators $\bM$ to the existence of \emph{an} operator $\bM$ or a \emph{certain class} of such operators satisfying a Robust IOP has the potential to yield IOP with better pseudo-norms.

As an example of an infinite-dimensional model with finite upper box-counting dimension, let's consider the problem experimented in \cite{Bourrier:2013wl}: we consider $E=L^1(\rn)\cap L^2(\rn)$ and aim at decoding a probability density $p\in E$ from a linear measurement $\bM p$. The \emph{a priori} on $p$ is that it can be expressed as a linear combination of a few densities taken in a set $\P$. In \cite{Bourrier:2013wl}, the authors considered $\P$ as a set of isotropic Gaussians, that is
\begin{equation}
\P=\left\{p_{\bm}: \bx\rightarrow \exp\left(-\ntwo{\bx-\bm}^2\right) |\bm\in\rn\right\}.
\end{equation}
Denoting $\Sigma_k(\P)$ the compact set of convex linear combinations of $k$ elements in $\P$ with $\ntwo{\bm}\leq C$, the upper-box counting dimension of $\Sigma_k(\P)$ is upper bounded by $k(n+1)$, so that a prevalent set of linear operators satisfies the IOP \eqref{eq:iop_topological} as soon as the number of measurements satisfies $m> 2k(n+1)$.

This example shows that one can obtain uniform IOP for an infinite-dimensional model which ``spans in an infinite number of directions'', such as the aforementioned model $\Sigma_k(\P)$. We hope that more precise characterizations on this kind of IOP can be obtained in this general framework.
}

\subsection{Upper-bound on the $M$-norm by an atomic norm}

As we have seen, provided a \emph{lower-RIP} on $\sms$, an NSP can be derived with the $M$-norm as $\ny{.}$. However, this may look like a tautology since the $M$-norm explicitly depends on $\bM$. Hence, one may wonder if this NSP is of any use. We will prove in the following that provided an \emph{upper-RIP} on a certain cone $\Sigmap$ (which can be taken as $\r\Sigma$), a more natural upper bound can be derived by bounding the $M$-norm with an atomic norm \cite{ChandrasekaranRPW12}. In particular, this type of inequality applied to the usual $k$-sparse vectors and low-rank matrices models give, under standard RIP conditions, instance optimality upper bounds with typical norms. 

We will suppose in this section that $\nx{.}$ is a norm.

\subsubsection{The atomic norm $\spnorm{.}$}

Let $\Sigmap$ be a subset of $E$ and let $E^{\prime}$ be the closure of $\mathrm{span}(\Sigmap)$ with respect to the norm $\nx{.}$. For $\bx\in E^{\prime}$, one can define the ``norm'' $\spnorm{\bx}$ by:
\begin{equation}
\label{cone_norm}
\spnorm{\bx} := \inf{} {\left\{\sum_{k=0}^{+\infty}\nx{\bx_k}: \forall k, \bx_k\in\r\Sigmap\;\mathrm{and}\;\nx{\bx-\sum_{k=0}^K\bx_k}\rightarrow_{K\rightarrow +\infty} 0\right\}}.
\end{equation}

Remark that there may be some vectors $\bx$ for which $\spnorm{\bx}=+\infty$, if $\sum_{k=0}^{+\infty}\nx{\bx_k}=+\infty$ for any decomposition of $\bx$ as an infinite sum of elements of $\r\Sigmap$. However, the set $V=\{\bx\in E|\spnorm{\bx}<+\infty\}$ is a normed subspace of $E$ which contains $\Sigmap$ \cite{DeVT96}. In the following, we assume that $V=E$. Note that this norm can be linked to atomic norms defined in \cite{ChandrasekaranRPW12} by considering $\mathcal{A}$ as the set of normalized elements of $\Sigmap$ with respect to $\nx{.}$.

Now suppose $\bM$ satisfies an upper-RIP on $\Sigmap$, so that
\begin{equation}
\forall\bxp\in\Sigmap, \nz{\bM\bxp}\leq \beta\nx{\bxp}.
\end{equation}

For $\bx\in E$ admitting a decomposition $\sum_{k=0}^{+\infty}\bx_k$ on $\r\Sigmap$, we can therefore upper bound $\nz{\bM\bx}$ by $\sum_{k=0}^{+\infty}\nz{\bM\bx_k}\leq \beta \sum_{k=0}^{+\infty}\nx{\bx_k}$. This inequality is valid for any decomposition of $\bx$ as a sum of elements of $\r\Sigmap$, so that $\nz{\bM\bx}\leq \beta\spnorm{\bx}$. Therefore, under these hypotheses,
\begin{equation}
\forall\bx\in E, \nM{\bx}\leq \nx{\bx}+\frac{\beta}{\alpha}\spnorm{\bx}\leq\left(1+\frac{\beta}{\alpha}\right)\spnorm{\bx}.
\end{equation}

In particular, we have the following result:

\begin{theorem}
\label{thm:rip_nsp_full}
Suppose $\bM$ satisfies the lower-RIP on $\sms$ with constant $\alpha$ and the upper-RIP on $\Sigma$ with constant $\beta$, that is
\begin{equation}
\forall\bx\in \sms, \alpha\nx{\bx}\leq \nz{\bM\bx}
\end{equation}
and
\begin{equation}
\forall\bx\in \Sigma, \nz{\bM\bx}\leq \beta\nx{\bx}.
\end{equation}
Then for all $\delta>0$, there exists a decoder $\dd$ satisfying
\begin{equation}
\forall\bx\in E,\forall\be\in F, \nx{\bx-\dd(\bM\bx+\be)}\leq 2\left(1+\frac{\beta}{\alpha}\right) d_{\Sigma}(\bx,\Sigma)+\frac{2}{\alpha}\ny{\be}+\delta,
\end{equation}
where $d_{\Sigma}$ is the distance associated to the norm $\snorm{.}$.
\end{theorem}

\begin{remark}
Note that these results can be extended with relative ease to the case where $\nx{.}$ is not necessarily homogeneous but $p$-homogeneous, that is $\nx{\lambda\bx}=|\lambda|^p\nx{\bx}$.
\end{remark}

\subsubsection{Study of $\snorm{.}$ in two usual cases}

We now provide a more thorough analysis of the norm $\snorm{.}$ for usual models which are sparse vectors and low-rank matrices. In particular, we give a simple equivalent of this norm involving usual norms in the case where $\nx{.}=\ntwo{.}$ (for matrices, this is the Frobenius norm).

The norm $\snorm{.}$ relies on the decomposition of a vector as a sum of elements of $\r\Sigma$. When $\Sigma$ is the set of $k$-sparse vectors or the set or matrices of rank $k$, there are particular decompositions of this type:
\begin{itemize}
\item In the case where $\Sigma$ is the set of $k$-sparse vectors, a vector $\bx$ can be decomposed as $\sum_{j=1}^{\infty}\bx_j$, where all $\bx_j$ are $k$-sparse vectors with disjoint supports, which are eventually zero, and such that any entry of $\bx_j$ does not exceed any entry of $\bx_{j-1}$ (in magnitude). This is a decomposition of $\bx$ into disjoint supports of size $k$ with a nonincreasing constraint on the coefficients.
\item Similarly, in the case where $\Sigma$ is the set of matrices of rank $k$ and $\bN$ is a matrix, the SVD of $\bN$ gives a decomposition of the form $\bN=\sum_{j=1}^{\infty}\bN_j$, where the $\bN_j$ are rank $k$, eventually zero matrices such that any singular value of $\bN_j$ does not exceed any singular value of $\bN_{j-1}$.
\end{itemize}

For $j\geq 2$, we can upper bound the quantity $\ntwo{\bx_j}$ using the assumption on the particular decomposition: $\ntwo{\bx_j}\leq \sqrt{k}\ninf{\bx_j}\leq \sqrt{k}\frac{\none{\bx_{j-1}}}{k}=\frac{\none{\bx_{j-1}}}{\sqrt{k}}$. Similarly, $\ntwo{\bN_j}\leq \frac{\tracenorm{\bN_{j-1}}}{\sqrt{k}}$, where $\tracenorm{.}$ is the trace norm, defined as the sum of singular values. We can therefore, in both cases, upper bound the norm $\snorm{.}$. In the case of $k$-sparse vectors, this gives:
\begin{equation}
\snorm{\bx}\leq \ntwo{\bx_1} + \sum_{j\geq 1}\frac{\none{\bx_j}}{\sqrt{k}}\leq \ntwo{\bx} + \frac{\none{\bx}}{\sqrt{k}}.
\end{equation}

In the case of matrices of rank $k$, this gives:
\begin{equation}
\snorm{\bN}\leq \ntwo{\bN_1} + \sum_{j\geq 1}\frac{\none{\bN_j}}{\sqrt{k}}\leq \nfrob{\bN} + \frac{\tracenorm{\bN}}{\sqrt{k}}.
\end{equation}

We can also upper bound the right hand side of these equations by $\mathcal{O}(\snorm{.})$ with a small constant, which will prove that the norms defined in these equations are of the same order. Indeed, a simple application of the triangle inequality gives us first that $\ntwo{\bx}\leq \snorm{\bx}$ and $\nfrob{\bN}\leq \snorm{\bN}$. Then, considering a decomposition of $\bx$ as a sum of $k$-sparse vectors $\sum_{j\geq 1}\bx_j$, we get 
\begin{equation}
\frac{\none{\bx}}{\sqrt{k}}\leq \sum_{j\geq 1} \frac{\none{\bx_j}}{\sqrt{k}}\leq \sum_{j\geq 1}\ntwo{\bx_j}
\end{equation}
(indeed, each $\bx_j$ can be viewed as a $k$-dimensional vector and we have for such a vector $\none{\bx_j}\leq \sqrt{k}\ntwo{\bx_j}$). Similarly,
\begin{equation}
\frac{\tracenorm{\bN}}{\sqrt{k}}\leq \sum_{j\geq 1}\nfrob{\bN_j}.
\end{equation}

Since these upper bounds are satisfied for any decomposition, they can be replaced respectively by $\snorm{\bx}$ and $\snorm{\bN}$. Finally, we have
\begin{equation}
\begin{array}{ccc}
\ntwo{\bx} + \frac{\none{\bx}}{\sqrt{k}}\leq 2\snorm{\bx} & \mathrm{and} & \nfrob{\bN} + \frac{\tracenorm{\bN}}{\sqrt{k}}\leq 2\snorm{\bN}.
\end{array}
\end{equation}
We have thus shown:
\begin{lemma}
When $\Sigma$ is the set of $k$-sparse vectors, the norm $\snorm{.}$ satisfies 
\begin{align}
 \|\cdot\|_{\Sigma} & \leq \ntwo{\cdot}+\frac{\none{\cdot}}{\sqrt{k}}  \leq 2 \|\cdot\|_{\Sigma}.\\
\intertext{When $\Sigma$ is the set of rank-$k$ matrices, it satisfies}
 \|\cdot\|_{\Sigma} & \leq \nfrob{\cdot}+\frac{\tracenorm{\cdot}}{\sqrt{k}} \leq 2 \|\cdot\|_{\Sigma}.
\end{align}
\end{lemma}
\paragraph{} We can thus remark that for these two standard models, the norm $\snorm{.}$ can easily be upper bounded by usual norms under RIP conditions, yielding an IOP with a usual upper bound. We can also note that stronger RIP conditions can yield a stronger result: in \cite{Candes08}, the author proves that under upper and lower-RIP on $\sms$ with $\Sigma$ being the set of $k$-sparse vectors, an instance optimal decoder can be defined as the minimization of a convex objective: the $\ell^{1}$ norm, which appears as strongly connected to the norm $\|\cdot\|_\Sigma$. One may then wonder if a generalization of such a result is possible: when can an instance optimal decoder be obtained by solving a convex minimization problem with a norm related to $\|\cdot\|_{\Sigma}$?

\section{Discussion and outlooks on Instance Optimality}
\label{sec:discussion}

Let's now summarize the results and give some insights on interesting future work. As has been detailed throughout the paper, Instance Optimality is a property presenting several benefits:
\begin{itemize}
\item It can be defined in a very general framework, for any signal space, signal model and pseudo-norms, as well as for both noiseless and noisy settings.
\item It is a nice uniform formulation of the ``good behavior'' of a decoder and thus of the well-posedness of an inverse problem.
\item It can be linked to Null Space Property and Restricted Isometry Property, which provide necessary and/or sufficient conditions for the existence of an Instance Optimal decoder.
\end{itemize}

We now present some immediate outlooks and interesting open questions related to instance optimality and to the results presented in this paper.

\paragraph{Condition for the well-posedness of the ``optimal'' decoder.}

We have seen that for general models $\Sigma$, an additionnal term $\delta$ appears in the right hand side term of the instance optimality inequality (\eqref{io_generalized_nonexact},\eqref{io_gen_robust}), reflecting the fact that the minimal distance of the ``optimal'' decoder \eqref{decoder_exact} may not be reached at a specific point. However, as mentioned in Property \ref{prop:sigma_n_closed}, this additive constant can be dropped in the noiseless case provided $\Sigma+\N$ is a closed set. One can then wonder if there exists a similar condition (e.g., a sort of local compactness property) in the noisy case for which one can drop the constant $\delta$ and get a more usual instance optimality result.

\paragraph{Compressed graphical models.}

As has been mentioned in Section \ref{subsubsec:limits_dimred}, the case where $\Sigma$ is the set of symmetric definite positive square matrices with sparse inverse is related to high-dimensional Gaussian graphical models. In Lemma \ref{lem:sdp_mat}, we showed this type of models fits in our theory since we could apply Theorem \ref{thm:l2l2_io} in this case, proving the impossibility of $\ell^2/\ell^2$ IOP in a dimension-reduction case. Yet, as for other signal models, can Gaussian graphical models satisfy some IOP/NSP with different norms in a compressive framework?

\paragraph{Guarantees for signal-space reconstructions and more.} % Best $k$-term approximation of the gradient}

When $\bD$ is a redundant dictionary of size $d\times n$ and the signals of interest are vectors of the form $\bz=\bD\bx$, where $\bx$ is a sparse vector, traditional reconstruction guarantees from $\by = \bM \bz$ assume the RIP on the matrix $\bM\bD$. This is often too restrictive: for example when $\bD$ has strongly correlated columns, failure to identify $\bx$ from $\by$ does not necessarily prevent one from correctly estimating $\bz$. Recent work on {\em signal-space} algorithms  \cite{DavenportNW13} has shown that the $\bD$-RIP assumption on $\bM$ is in fact sufficient.  

The framework presented in this paper offers two ways to approach this setting:
% to study instance optimality for a certain decoder $\Delta$, upper bounding the reconstruction error $\|\bz-\Delta(\bM\bz)\|$:
\begin{itemize}
\item Considering $\Sigma=\Sigma_k$ as the set of $k$-sparse vectors of dimension $n$ and $\bA=\bD$, the upper bound on the reconstruction error is of the form $d_{E}(\bx,\Sigma_k)$. Signal-space guarantees can be envisioned by choosing a metric $\|\cdot\|_{E} =  \|\bD \cdot\|$.
\item Considering $\Sigma=\bD\Sigma_k$ as the set of $d$-dimensional vectors that have a $k$-sparse representation in the dictionary $\bD$ and $\bA=\bI$, the upper bound is of the form $d^{\prime}(\bz,\bD\Sigma_k)$.
\end{itemize}
In \cite{NeedellW13}, the authors propose a result similar to instance optimality by upper bounding, for a Total Variation decoder, the reconstruction error of an image $\bX$ from compressive measurement by a quantity involving $d_1(\nabla\bX,\Sigma_k)$, where $\nabla$ is the gradient operator, $\Sigma_k$ the $k$-sparse union of subspaces (in the gradient space) and $d_1$ is the $\ell^1$ distance. This quantity is therefore the distance between the gradient of the image and the $k$-sparse vectors model. Can such a bound be interpreted in our framework, and possibly be generalized to other types of signals?

\paragraph{Task-oriented decoders versus general purpose decoders.}

We already mentioned two very different application set-ups, in medical imaging and audio source separation, where only parts of the original signals need to be recovered. One can think of other, more dramatic, cases where only task-oriented linear features should be reconstructed. One such situation is met
in current image classification work-flows. Indeed, most recent state-of-art image classification methods rely on very high-dimensional image representation (e.g., so called Fisher vectors, of dimension ranging from 10,000 to 200,000) and conduct supervised learning on such labeled signals by means of linear SVMs \cite{Sanchez13}. Not only this approach yields top-ranking performance in terms of classification accuracy on challenging image classification benchmarks, but it also permits very large scale learning thanks to the low complexity of linear SVM training and its efficient implementations, e.g., with stochastic gradient descent. For each visual category to recognize, a linear classifier $\mbf{w}$ is learned, which associates to an input image with representation $\bx$ the score $\mbf{w}^T\bx$. The single or multiple labels that are finally assigned to $\bx$ by the system depend on the scores provided by all trained classifiers (typically from 10 to 100), hence on a vector of the form $\bA\bx$, where each row of $\bA$ is one one-vs-all linear SVM. In this set-up, the operator $\bA$ implies a dramatic dimension reduction. For very large scale problems of this type, storing and manipulating original image signatures in the database can become intractable. The theoretical framework proposed in this paper might help designing new solutions to this problem in the future. In particular, it will provide tools to answer the following questions:
\begin{itemize}
\item $\bA$ being given (learned on a labeled subset of the database): can one design a compressive measurement operator $\bM$ such that the ``classifiers'' scores can be recovered directly from the compressed image signature $\bM\bx$, hence avoiding the prior reconstruction of the high-dimensional signal $\bx$? 
\item $\bM$ being given (``legacy'' compressed storing of image signatures): what are the linear classifier collections that can be precisely emulated in the compressed domain thanks to a good decoder $\Delta$?   
\end{itemize}       
Note that this classification-oriented set-up might call for a specific norm $\gnorm{.}$ on the output of a linear score bank. 

Another important domain of application that might benefit from both aspects (general purpose and task-oriented) of our work is data analysis under privacy constraints. Two scenarii can be envisioned, where our framework could help decide whether or not such constraints are compatible with the analysis of interest:
\begin{itemize}
\item \emph{General purpose scenario}: given a linear measurement operator $\bM$ of interest for further analysis, can one guarantee that there is no decoder permitting good enough recovery of original signals? 
\item \emph{Task-oriented scenario}: the operator $\bM$ serving as a means to obfuscate original signals such that critical information can't be recovered, let's consider a specific analysis task on original signals requiring the application of the feature extractor $\bA$. Can this task be implemented on obfuscated signals instead, via a good decoder $\Delta$, hence in a privacy-preserving fashion?        
\end{itemize}

\paragraph{Worst case versus average case instance optimality.}

The raw concept of Instance Optimality has a major drawback: the uniformity of the bound may impose, in some settings, a large global instance optimality constant whereas the inverse problem is well posed for the vast majority of signals. Let's consider the example depicted in Figure \ref{fig:io_limit}, where the signal space $E$ is of dimension 2, the signal model $\Sigma$ is a point cloud mostly concentrated along the line $\mathcal{D}$ and the measurement operator $\bM$ is the orthogonal projection on $D$. The figure depicts the ratio (approximation error)/(distance to model) for each $\bx\in\r^2$. The optimal constant, which is the supremum of these ratios, is infinite: the ratio actually goes to infinity in the vicinity of the point $p$. However, for the vast majority of vectors, the ratio is rather low (the blue section covers most of the space).

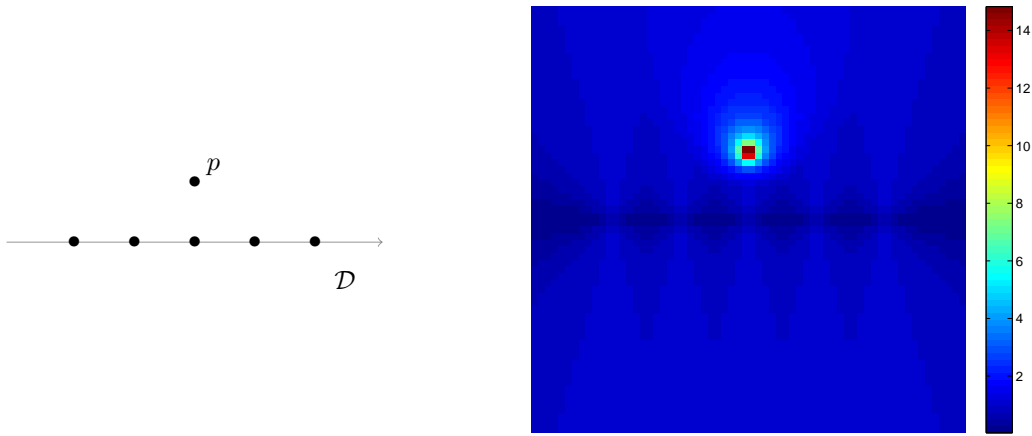
\begin{figure}[t]
\centering
\begin{minipage}{.45\linewidth}
\begin{tikzpicture}
%\draw (0,0.48) node{\includegraphics[scale=.5]{io_limit_model.pdf}};
\draw[->,opacity=.4] (-2.5,0)--(2.5,0);
\draw (0,0) node{\textbullet};
\draw (0,0.8) node{\textbullet};
\draw (.8,0) node{\textbullet};
\draw (1.6,0) node{\textbullet};
\draw (-.8,0) node{\textbullet};
\draw (-1.6,0) node{\textbullet};
\draw (2,-.5) node{$\mathcal{D}$};
\draw (0.25,1) node{$p$};
\end{tikzpicture}
\end{minipage}
\begin{minipage}{.45\linewidth}
\includegraphics[scale=.5]{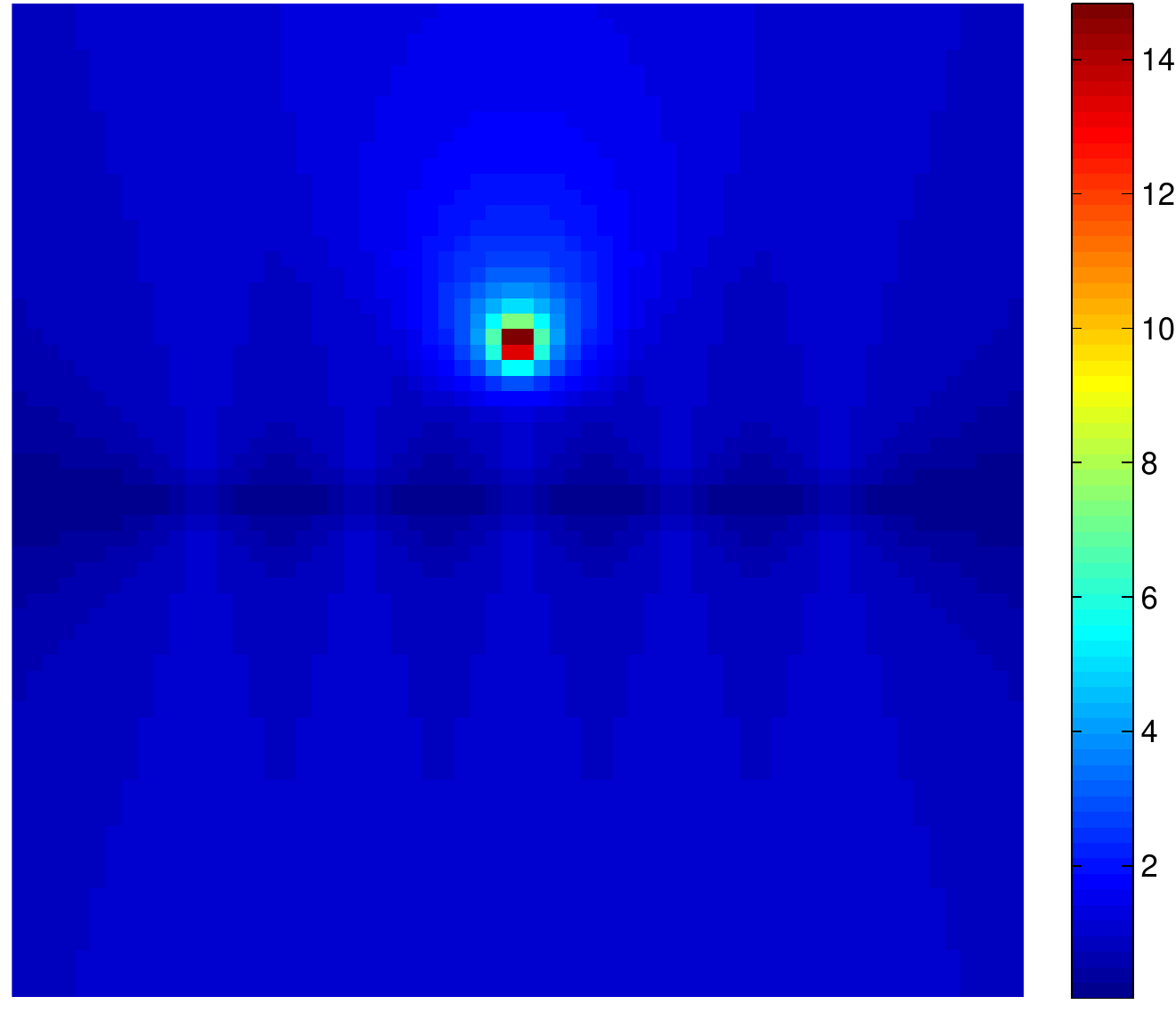}
\end{minipage}

\caption{Drawback of uniform instance optimality in a simple case: the model $\Sigma$ (\emph{Left}) is the set of black points including those on $D$ and the point $p$ and the operator $\bM$ is the 1-dimensional orthogonal projection on the horizontal axis. If we choose $\Delta$ as the pseudo-inverse of $\bM$, the depicted IO ratio (\emph{Right}) is low on most of the space, but the uniform constant is infinite.}
\label{fig:io_limit}
\end{figure}

An interesting outlook to circumvent this pessimistic ``worst-case'' phenomenon is to consider a probabilistic formulation of instance optimality, as in \cite{Cohen09compressedsensing}: given $\Omega$ a probability space with probability measure $P$, and considering $\bM$ as a random variable on $\Omega$, is there a decoder  $\Delta(.|\bM)$ (which computes an estimate given the observation {\em and} the particular draw of the measurement operator $\bM$) such that  for any $\bx\in E$, the instance optimality inequality
\begin{equation}
\nx{\bx-\Delta(\bM\bx|\bM)}\leq C d_E(\bx,\Sigma)
\end{equation}
holds with high probability on the drawing of $\bM$?  A particular challenge would be to understand in which dimension reduction scenarii there exists both a probability measure and a decoder with the above property.  Another possible formulation of probabilistic instance optimality is to define a probability distribution on the signal space and to upper bound the average reconstruction error of the vectors, as in \cite{YuSapiro2011}.

\section*{Acknoledgements}

This work was supported in part by the European Research Council,
PLEASE project (ERC-StG-2011-277906). The authors also wants to thank the anonymous reviewers for their remarks about instance optimality in infinite dimensions and for providing the example given in Section \ref{sec:sparse_hilbert}.

\bibliographystyle{plain}
\bibliography{egbib}

\appendix

\section{Well-posedness of the finite UoS decoder}
\label{sec:sparse_decoder}

In this section, we will prove that if $\Sigma$ is a finite union of subspaces in $\rn$ and $\|.\|$ a norm on $\rn$, then the quantity $\arg\min_{\bz\in(\bx+\N)}{d(\bz,\Sigma)}$, where $d$ is the distance relative to $\|.\|$, is defined for all $\bx\in\rn$.

Let's first prove the following lemma:

\begin{lemma}
\label{lem:min_dist_spaces}
Let $V$ and $W$ be two subspaces of $\rn$ and $\|.\|$ a norm on $\rn$. Then $\forall\bx\in\rn, \exists\by\in (\bx+V)$ such that $d(\by,W)=d(\bx+V,W)$, where $d$ is the distance derived from $\|.\|$.
\end{lemma}

\begin{proof}
Let $\Phi$ be defined on $V+W$ by $\Phi(\bu)=\|\bu-\bx\|$. Since $\Phi(\bu)\geq \|\bu\|-\|\bx\|$, we have $\lim_{\|\bu\|\rightarrow +\infty}\Phi(\bu)=+\infty$, so that $\exists M>0$ such that $\|\bu\|>M\Rightarrow \Phi(\bu)\geq \|\bx\|$. The set $B=\{\bu\in V+W, \|\bu\|\leq M\}$ is a closed ball of $V+W$ and is thus a compact. Since $\Phi$ is continuous, $\Phi$ has a minimizer $\bv$ on $B$. $0\in B$, so that $\Phi(0)=\|\bx\|\geq \Phi(\bv)$. For all $\bu$ such that $\|\bu\|>M$, we have $\Phi(\bu)\geq \|\bx\|\geq \Phi(\bv)$, so that $\bv$ is a global minimizer of $\Phi$.

We therefore have $\forall (\bu,\bw)\in V\times W, \|\bx-\bv\|\leq \|\bx-(\bu+\bw)\|$. The vector $\bv$ can be written $\bf+\bg$ with $\bf\in V$ and $\bg\in W$, so that the vector $\by = \bx-\bf$, which belongs to $\bx +V$, satisfies $d(\bx-\bf,W)=\|(\bx-\bf)-\bg\|=d(\bx,V+W)=d(\bx+V,W)$, which proves the result.

\end{proof}

Let $\Sigma=\underset{i\in\llbracket 1,p\rrbracket }{\cup}V_i$, where $V_i$ are subspaces of $\rn$. Lemma \ref{lem:min_dist_spaces} applied to $V=\N$ and $W=V_i$ ensures the existence of $\bx_i\in(\bx+\N)$ such that $d_E(\bx_i,V_i)=d_E(\bx+\N,V_i)$. Therefore, $\Delta(\bM\bx)$ can be defined as $\argmin{\{\bx_i,i\in\llbracket 1,p\rrbracket\}}{d_E(\bx_i,V_i)}$ and satisfies $d_E(\Delta(\bM\bx),\Sigma)=d_E(\bx+\N,\Sigma)$, so that the decoder $\Delta(\bM\bx)=\argmin{\bz\in(\bx+\N)}{d(\bz,\Sigma)}$ is properly defined. In particular, this applies to the decoder \eqref{decoder_initial}.

\section{Proof of Theorem \ref{thm:io_nsp}}

Let $\delta >0$ and $\dd$ and $C$ be such that \eqref{io_generalized_nonexact} holds $\forall\bx\in E$. Let $\bh\in\N$. Then $\exists \bh_0\in\sms$ such that $d_E(\bh,\bh_0)\leq d_E(\bh,\sms)+\delta$. Let $\bh_0=\bh_1 - \bh_2$ with $\bh_1,\bh_2\in\Sigma$, and $\bh_3=\bh-\bh_0$. Since $\bh\in\N$, we have:
\begin{equation}
\label{ker_relation}
\bM(\bh_1+\bh_3)=\bM\bh_2.
\end{equation}
Applying \eqref{io_generalized_nonexact} to $\bx=\bh_2\in\Sigma$ and using the fact that $\ny{0}=0$, we get:

\begin{equation}
\label{h2_relation}
\nx{\bA\bh_2-\dd(\bM\bh_2)}\leq\delta.
\end{equation}
Let's now find an upper bound for $\nx{\bA\bh}$:
\begin{align}
\label{upper_bound_first}
\nx{\bA\bh}&=\nx{\bA(\bh_1-\bh_2+\bh_3)}\nonumber\\
&=\nx{\bA(\bh_1+\bh_3)-\dd(\bM(\bh_1+\bh_3))-\bA\bh_2+\dd(\bM(\bh_1+\bh_3))}\nonumber\\
&\leq \nx{\bA(\bh_1+\bh_3)-\dd(\bM(\bh_1+\bh_3))}+\nx{\bA\bh_2-\dd(\bM(\bh_1+\bh_3))},
\end{align}
where we have used \eqref{nx_symmetry} and \eqref{nx_triangle} for the last inequality.
Combining (\ref{ker_relation}) and (\ref{h2_relation}), we get that:
\begin{equation}
\label{second_term}
\nx{\bA\bh_2-\dd(\bM(\bh_1+\bh_3))}\leq\delta.
\end{equation}
Applying \eqref{io_generalized_nonexact} to $\bx=\bh_1+\bh_3$, we get:
\begin{align}
\label{first_term}
\nx{\bA(\bh_1+\bh_3)-\dd(\bM(\bh_1+\bh_3))}&\leq C\dy{\bh_1+\bh_3,\Sigma}+\delta
\leq C\ny{\bh_3}+\delta\nonumber\\
&= C\dy{\bh,\bh_0}+\delta\leq C\dy{\bh,\sms}+(C+1)\delta.
\end{align}
Combining \eqref{upper_bound_first}, \eqref{second_term} and \eqref{first_term} gives:
\begin{equation}
\label{nsp_eps}
\nx{\bA\bh}\leq C\dy{\bh,\sms}+(C+2)\delta.
\end{equation}
\eqref{nsp_eps} is valid for all $\delta >0$, so it is valid for $\delta=0$. This gives us the property \eqref{nsp_generalized} with $D=C$.

\section{Proof of Theorem \ref{thm:nsp_io}}

Let's first assume that \eqref{exact_decoder} holds and define the following decoder on $F$:
\begin{equation}
\label{decoder_exact}
\deltap(\bM\bx)=\argmin{\bz\in(\bx+\N)}{d_E(\bz,\Sigma)}.
\end{equation}
Note that the decoder is well defined, since $\bM\bx_1=\bM\bx_2\Rightarrow \bx_1+\N=\bx_2+\N$.

For $\bx\in E$, we have $\bx-\deltap(\bM\bx)\in\N$, so that \eqref{nsp_generalized} yields:
\begin{align}
\nx{\bA\bx-\bA\deltap(\bM\bx)}&\leq Dd_E(\bx-\deltap(\bM\bx),\sms)\nonumber\\
&\leq Dd_E(\bx,\Sigma) + Dd_E(\deltap(\bM\bx),\Sigma)\nonumber\\
&\leq 2Dd_E(\bx,\Sigma),
\end{align}
where we have used \eqref{nx_triangle} for the second inequality.
The last inequality comes from \eqref{decoder_exact}, which yields $d_E(\deltap(\bM\bx),\Sigma)\leq d_E(\bx,\Sigma)$. Therefore, by posing $\d=\bA\deltap$, we get \eqref{io_generalized}.

\bigskip

Let's return to the general case, and consider $\nu>0$. We define the following decoder on $F$:
\begin{equation}
\label{decoder_approx}
\deltap_{\nu}(\bM\bx)\in\{\bu\in (\bx+\N)| d_E(\bu,\Sigma)\leq d_E(\bx+\N,\Sigma)+\nu\}.
\end{equation}
Note that this set may not contain a unique element and thus this definition relies on the axiom of choice.

For $\bx\in E$, we have again $\bx-\deltap_{\nu}(\bM\bx)\in\N$, so that by \eqref{nsp_generalized}:
\begin{align}
\nx{\bA\bx-\bA\deltap_{\nu}(\bM\bx)}&\leq Dd_E(\bx-\deltap_{\nu}(\bM\bx),\sms)\nonumber\\
&\leq Dd_E(\bx,\Sigma) + Dd_E(\deltap_{\nu}(\bM\bx),\Sigma)\nonumber\\
&\leq 2Dd_E(\bx,\Sigma)+D\nu,
\end{align}
where we have used \eqref{nx_triangle} again for the second inequality.
The last inequality comes from \eqref{decoder_approx}, which yields $d_E(\deltap_{\nu}(\bM\bx),\Sigma)\leq d_E(\bx,\Sigma)+\nu$. Therefore, by posing $\dd=\bA\deltap_{\delta/D}$, we get \eqref{io_generalized_nonexact}.
\section{Proof of Proposition \ref{prop:sigma_n_closed}}
Let $\bx\in E$ and $\nu >0$. If $0=d_E(\bx+\N,\Sigma)=d_E(\bx,\Sigma+\N)$, then since $\Sigma+\N$ is a closed set, $\bx\in\Sigma+\N$, and therefore $(\bx+\N)\cap \Sigma\ne\emptyset$. In this case, we define $\deltap_{\nu}(\bM\bx)$ as any element of $(\bx+\N)\cap \Sigma$.

If $d_E(\bx+\N,\Sigma)>0$, then we define $\deltap_{\nu}(\bM\bx)\in\{\bu\in (\bx+\N)| d_E(\bu,\Sigma)\leq (1+\nu)d_E(\bx+\N,\Sigma)\}$. This provides a consistent definition of $\deltap_{\nu}$.

Let's remark that for all $\bx\in E$, $d_E(\deltap_{\nu}(\bM\bx),\Sigma)\leq (1+\nu)d_E(\bx+\N,\Sigma)\leq (1+\nu)d_E(\bx,\Sigma)$.

For $\bx\in E$, $\bx-\deltap_{\nu}(\bM\bx)\in\N$, so that \eqref{nsp_generalized} gives:
\begin{align}
\nx{\bA\bx-\bA\deltap_{\nu}(\bM\bx)}&\leq Dd_E(\bx-\deltap_{\nu}(\bM\bx),\sms)\nonumber\\
&\leq Dd_E(\bx,\Sigma) + Dd_E(\deltap_{\nu}(\bM\bx),\Sigma)\nonumber\\
&\leq (2+\nu)Dd_E(\bx,\Sigma).
\end{align}
Defining $\dd=\bA\deltap_{\nu}$, we get the desired result.

\section{Proof of Theorem~\ref{thm:io_nsp_noise} and Theorem~\ref{thm:io_nsp_noise_aware}}

%Suppose $\forall\epsilon,\delta>0$, there exists a decoder $\d_{\epsilon,\delta}$ satisfying \eqref{io_gen_robust_noise_aware}:
%%\begin{equation}
%\[
%\forall\bx\in E,\forall\be\in F,\ N_{Z}(\be) \leq \epsilon \Rightarrow \nx(\bA\bx-\dd(\bM\bx+\be))\leq C_1 d_E(\bx,\Sigma)+ C_2\epsilon+\delta.
%\]

%%%%%%

Let's first remark that applying~\eqref{io_gen_robust} (resp.~\eqref{io_gen_robust_noise_aware}) with $\bx=\bz\in\Sigma$ and $\be=0$ yields $\nx{\bA\bz-\dd(\bM\bz)}\leq\delta$ and $\nx{\bA \bz-\d_{\delta,\epsilon}(\bM\bz)} \leq C_{2}\epsilon+\delta$ for any $\bz\in\Sigma$, $\epsilon\geq0$, where we have used the fact that $\nz{0}=0$.

Let $\bh\in E$ and $\bz\in\Sigma$. We apply~\eqref{io_gen_robust} (resp.~\eqref{io_gen_robust_noise_aware}) with $\bx=\bz-\bh$, $\be=\bM\bh$, and $\epsilon = \nz{\bM\bh}$, which yields:
\begin{eqnarray*}
\nx{\bA\bz-\bA\bh-\dd(\bM\bz)}&\leq& C_1\dy{\bz-\bh,\Sigma} + C_2\nz{\bM\bh}+\delta.\\
\nx{\bA\bz-\bA\bh-\d_{\delta,\epsilon}(\bM\bz)}&\leq& C_1\dy{\bz-\bh,\Sigma} + C_2\nz{\bM\bh}+\delta.\\
\end{eqnarray*}

Let's remark that \eqref{nx_symmetry} and \eqref{nx_triangle} imply $\nx{\by}\leq \nx{\bx-\by}+\nx{\bx}$ for all $\bx,\by\in G$. Therefore,
since $\nx{\bA\bz-\dd(\bM\bz)}\leq\delta$ (resp.  $\nx{\bA\bz-\d_{\delta,\epsilon}(\bM\bz)}\leq C_2\nz{\bM\bh}+\delta$), we have:
\begin{eqnarray*}
\nx{\bA\bh}&\leq& C_1\dy{\bz-\bh,\Sigma} + C_2\nz{\bM\bh}+2\delta.\\
(\textrm{respectively}\ \nx{\bA\bh}&\leq& C_1\dy{\bz-\bh,\Sigma} + 2C_2\nz{\bM\bh}+2\delta.)
\end{eqnarray*}

This last inequality is valid for all $\bz\in\Sigma$, therefore~\eqref{io_gen_robust} implies:
\begin{align}
\nx{\bA\bh}&\leq C_1\inf{\bz\in\Sigma}{\dy{\bz-\bh,\Sigma}}+C_2\nz{\bM\bh}+2\delta \nonumber \\
&= C_1\inf{\bz\in\Sigma}{\inf{\bu\in\Sigma}{\ny{\bz-\bh-\bu}}}+C_2\nz{\bM\bh}+2\delta \nonumber \\
&= C_1\dy{\bh,\sms}+C_2\nz{\bM\bh}+2\delta,\label{eq1}
\end{align}
where we have used \eqref{nx_symmetry} for the last inequality. Similarly,~\eqref{io_gen_robust_noise_aware} implies
\begin{align}
\nx{\bA\bh}&\leq C_1\dy{\bh,\sms}+2C_2\nz{\bM\bh}+2\delta.\label{eq2}
\end{align}
We conclude by using the fact that~\eqref{eq1} and~\eqref{eq2} hold for all $\delta>0$.
%, we have:
%\begin{equation}
%\nx(\bA\bh)\leq C_1\dy{\bh,\sms}+C_2\nz(\bM\bh),
%\end{equation}
%which proves the result with $D_1=C_1$ and $D_2=C_2$.

\section{Proof of Theorem \ref{thm:nsp_io_noise}}

Let's suppose \eqref{robust_nsp} and define for $\delta>0$ the decoder $\Delta_{\delta}^{\prime}:F\rightarrow E$ such that $\forall \by\in F:$
\begin{equation}
\label{decoder_robust}
D_1\dy{\dd^{\prime}(\by),\Sigma}+D_2\dz{\bM\dd^{\prime}(\by),\by} \leq \inf{\bu\in E}{\left[D_1\dy{\bu,\Sigma}+D_2\dz{\bM\bu,\by}\right]}+\delta.
\end{equation}
Let's prove that this decoder meets property \eqref{io_gen_robust}.

Let $\bx\in E$ and $\be\in F$. Applying \eqref{robust_nsp} with $\bh=\bx-\dd^{\prime}(\bM\bx+\be)$, we get:
\begin{align}
\nx{\bA(\bx-\dd^{\prime}(\bM\bx+\be))}&\leq D_1\dy{\bx-\dd^{\prime}(\bM\bx+\be),\sms}+D_2\nz{\bM(\bx-\dd^{\prime}(\bM\bx+\be))}\nonumber\\
&\leq D_1\dy{\bx,\Sigma}+D_1\dy{\dd^{\prime}(\bM\bx+\be),\Sigma}\nonumber\\
&\quad +D_2\dz{\bM\dd^{\prime}(\bM\bx+\be),\bM\bx+\be}+D_2\nz{\be}\nonumber\\
&\leq 2D_1\dy{\bx,\Sigma}+2D_2\nz{\be}+\delta,
\end{align}
where we have used \eqref{nx_symmetry} and \eqref{nx_triangle} for the second inequality and the last inequality is a consequence of \eqref{decoder_robust}.

Posing $\dd=\bA\dd^{\prime}$ proves \eqref{io_gen_robust} with $C_1=2D_1$ and $C_2=2D_2$.

\section{Proof of Lemma \ref{lem:nsp_cone}}

The two equivalences are very similar to prove, so that we will only prove the first. \eqref{nsp_cone} $\Rightarrow$ \eqref{nsp_sms} is obvious. Let's now suppose \eqref{nsp_sms}, so that:
\begin{equation}
\forall \bh\in\N,\forall\bz\in\sms, \nx{\bh}\leq D\ny{\bh-\bz}.
\end{equation}
By homogeneity, we also have:
\begin{equation}
\forall \lambda\in\r^*, \forall \bh\in\N, \forall\bz\in\sms, \nx{\lambda\bh}\leq D\ny{\lambda\bh-\bz},
\end{equation}
so that:
\begin{equation}
\forall \lambda\in\r^*, \forall \bh\in\N, \forall\bz\in\sms, \nx{\bh}\leq D\ny{\bh-\bz/\lambda}.
\end{equation}
This last inequality yields \eqref{nsp_cone}.

\section{Proof of Theorem \ref{thm:l2l2_io}}

Let's note $\bMt= \bM_{|V}$ and $\Nt=\N\cap V$. Let $m$ be the dimension of the range of $\bMt$, so that $\Nt$ is of dimension $n-m$. Let $\bh_1,\ldots,\bh_{n-m}$ be an orthonormal basis of $\Nt$. We have:
\begin{equation}
\label{dim_upper_bound}
n-m=\sum_{j=1}^{n-m}\ntwo{\bh_j}^2\leq \frac{1}{K}\sum_{j=1}^{n-m}\sum_{i=1}^n\ps{\bh_j}{\bz_i}^2.
\end{equation}
Using \eqref{def_l2nsp_constant}, we get that, for all $\bh\in\N$ and unit-norm vector $\bz\in\sms$, $\ps{\bh}{\bz}^2\leq \left(1-\frac{1}{D_*^2}\right)\ntwo{\bh}^2$. If we denote by $\pnt$ the orthogonal projection on $\Nt$ and apply this inequality with $\bh=\pnt(\bz_i)=\sum_{j=1}^{n-m} \ps{\bh_j}{\bz_i}\bh_j$ and $\bz=\bz_i$, we get that $\ntwo{\pnt(\bz_i)}^4\leq \left(1-\frac{1}{D_*^2}\right)\ntwo{\pnt(\bz_i)}^2$, which can be simplified to $\ntwo{\pnt(\bz_i)}^2=\sum_{j=1}^{n-m}\ps{\bh_j}{\bz_i}^2\leq \left(1-\frac{1}{D_*^2}\right)$ even if $\ntwo{\pnt(\bz_i)}=0$.

Using this relation in \eqref{dim_upper_bound}, we get:
\begin{equation}
n-m\leq \frac{n}{K}\left(1-\frac{1}{D_*^2}\right),
\end{equation}
so that:
\begin{equation}
m\geq n\left(1-\frac{1}{K}\left(1-\frac{1}{D_*^2}\right)\right).
\end{equation}

We get the lower bound on $D_*^2$ by isolating it in the inequality.

\section{Proof of Theorem \ref{thm:nsp_mnorm}}

Let $\bh\in E$ and $\bz\in\sms$. We have the following inequalities:

\begin{equation}
\label{eq:upper_bound_h}
\nx{\bh} \leq \nx{\bh-\bz} + \nx{\bz} \leq \nx{\bh-\bz} + \frac{1}{\alpha}\nz{\bM\bz},
\end{equation}
where we have used the lower-RIP for the second inequality.

A similar consideration on $\bM\bz$ yields:
\begin{equation}
\label{eq:upper_bound_mz}
\nz{\bM\bz}\leq \nz{\bM(\bz-\bh)}+\nz{\bM\bh}.
\end{equation}

Substituting \eqref{eq:upper_bound_mz} into \eqref{eq:upper_bound_h}, we get:
\begin{align}
\nx{\bh} &\leq \nx{\bh-\bz}+\frac{1}{\alpha}\nz{\bM(\bh-\bz)}+\frac{1}{\alpha}\nz{\bM\bh}\nonumber\\
&= \nM{\bh-\bz} + \frac{1}{\alpha}\nz{\bM\bh}.
\end{align}

Taking the infimum of the right hand-side quantity over all $\bz\in\sms$, one gets the desired Robust NSP:
\begin{equation}
\nx{\bh}\leq \dM{\bh,\sms} + \frac{1}{\alpha}\nz{\bM\bh}.
\end{equation}

\end{document}